\documentclass{article}
\usepackage{amsmath,amssymb,amsthm,authblk}
\usepackage{program}
\catcode`\|=12\relax
\usepackage{tikz}

\title{Using Colors and Sketches to Count Subgraphs in a Streaming Graph}

\newcommand{\calX}{{\cal{X}}}
\newcommand{\calG}{{\cal{G}}}

\newcommand{\calZ}{{\cal{Z}}}
\newcommand{\calM}{{\cal{M}}}

\newcommand{\calC}{{\cal{C}}}
\newcommand{\calS}{{\cal{S}}}
\newcommand{\calP}{{\cal{P}}}
\newcommand{\calQ}{{\cal{Q}}}
\newcommand{\calV}{{\cal{V}}}
\newcommand{\calE}{{\cal{E}}}
\newcommand{\vect}[1]{\overrightarrow{#1}}

\newcommand{\auto}{{\rm auto}}

\newcommand{\tr}{{\rm tr}}
\newcommand{\Count}{{\rm Count}}

\theoremstyle{thmstyleone}
\newtheorem{lemma}{Lemma}
\newtheorem{theorem}{Theorem}

\begin{document}

\author{Shirin Handjani\thanks{shirin@ccr-lajolla.org}}

\author{Douglas Jungreis\thanks{jungreis@ccr-lajolla.org}}

\author{Mark Tiefenbruck\thanks{mgtiefe@ccr-lajolla.org}}

\affil{IDA Center for Communications Research, La Jolla}

\maketitle

\abstract{
Suppose we wish to estimate $\#H$, the number of copies of some small
graph~$H$ in a large streaming graph $G$.  There are many algorithms
for this task when $H$ is a triangle, but just a few that apply to
arbitrary~$H$.  Here we focus on one such algorithm, which was introduced
by Kane, Mehlhorn, Sauerwald, and Sun.  The storage and update time
per edge for their algorithm are both $O(m^k/(\#H)^2)$, where $m$ is
the number of edges in $G$, and $k$ is the number of edges in $H$.
Here, we propose three modifications to their algorithm that can
dramatically reduce both the storage and update time.  Suppose that
$H$ has no leaves and that $G$ has maximum degree $\leq m^{1/2 - \alpha}$,
where $\alpha > 0$.  Define $C = \min(m^{2\alpha},m^{1/3})$. Then in
our version of the algorithm, the update time per edge is $O(1)$, and
the storage is approximately reduced by a factor of $C^{2k-t-2}$,
where $t$ is the number of vertices in $H$; in particular, the
storage is $O(C^2 + m^k/(C^{2k-t-2} (\#H)^2))$.}
\section{Introduction}

Suppose that a large simple graph $G$ is presented as a stream of edge
insertions and deletions, and suppose that $H$ is a very small graph
(e.g., a small clique or cycle).  Our goal is to estimate $\#H$, the
number of copies of $H$ that appear in~$G$, where we are permitted
a single pass through the stream.  This problem has received a great
deal of attention, particularly in the case where $H$ is a triangle;
however, there are only a few known techniques that apply to arbitrary
$H$.  Here we focus on the technique that was developed in~\cite{KMSS,MMPS},
which we refer to as the [KMSS]-algorithm.

The [KMSS]-algorithm, which uses complex-valued linear sketches,
has many strengths: it applies to arbitrary $H$; it can be used in
distributed settings; it allows edge deletions; and it is extremely
efficient in a variety of situations, such as when $H$ is a star graph.
However, there are many situations where the algorithm is not practical.
Suppose $G$ has $m$ edges, and suppose $H$ has $k$ edges and $t$ vertices.
When the [KMSS]-algorithm produces a single estimate of $\#H$, that estimate
has variance $\Theta(m^k)$, so it is necessary to produce $O(m^k/(\#H)^2)$
estimates and average them.  The storage and update time per edge
are proportional to the number of estimates produced, and are
therefore both $O(m^k/(\#H)^2)$.


In this paper, we describe three modifications to the
[KMSS]-algorithm that greatly reduce both the storage
and update time per edge.  Suppose that $H$ is a connected
graph with no leaves.  Suppose also that the maximum degree of any
vertex in $G$ is $\Delta \leq m^{1/2 - \alpha}$, where $\alpha > 0$,
and define $C = \min(m^{1/3}, m^{2\alpha})$.  Then the storage
required by our algorithm is $O(C^2 + m^k/(C^{2k-t-2} (\#H)^2))$,
i.e., it has been reduced approximately by a factor of $C^{2k-t-2}$.
The update time per edge is $O(1)$.

The problem of counting copies of a small graph $H$ in a large graph
$G$ has been studied extensively.  It has many applications, as diverse
as community detection, information retrieval, and motifs in bioinformatics;
see for instance \cite{ACIKMS,CC,ELS,KMPT,T}.  Here we restrict to the case
where $G$ is given as a data stream, and our goal is merely to estimate
$\#H$, as opposed to computing $\#H$ exactly.  Most work on this problem
has addressed the case where $H$ is a triangle
\cite{AGM,BKS,BFKP,BFLMS,CGT,ERSU,FHKS,HS,JSP,JG,KP,KL,KHP,KMT,MVV,PT,PTTW}.
A few authors have addressed other specific subgraphs, such as
butterflies~\cite{SSTZ} and cycles~\cite{MMPS}.
We are only aware of a few algorithms that apply to arbitrary
subgraphs~\cite{AKK,BC,KMSS,KKP}.  Two of these, \cite{BC}
and~\cite{AKK}, require multiple passes through the stream,
which we do not allow here.  The third, \cite{KMSS}, presents
the [KMSS]-algorithm, which is the focus of this paper.  The
last, \cite{KKP}, presents a vertex-sampling algorithm which,
in some situations, is extremely efficient, requiring storage
$O(m/(\#H)^{1/\tau})$, where $\tau$ is the fractional vertex
cover number of $H$.  However, this bound requires a strong
assumption on $G$: it either requires that $G$ have bounded
degree, or it requires that the maximum degree in $G$ is
$(\#H)^{1/(2\tau)}$ and that some optimal fractional vertex
cover of $H$ can place non-zero degree on every vertex.

In order to explain our contribution to this problem, we first need
to briefly review the [KMSS]-algorithm.  Consider a fixed $H$.  Many
independent estimates are made for $\#H$, and they are then averaged.
To get a single estimate, the first step is to arbitrarily assign
directions to the edges of $H$.  We refer to the resulting digraph
as $\vec{H}$ and its edges as
$\vect{a_1 a_2}, \vect{a_3 a_4}, \dots, \vect{a_{2k-1} a_{2k}}$.
Also each edge $vw$ of $G$ is replaced by two directed edges $\vect{vw}$
and $\vect{wv}$.  We refer to the resulting directed version of $G$ as
$\vec{G}$.  For any graph or digraph $X$, we refer to its vertices
and edges as $\calV(X)$ and $\calE(X)$.  Now we define $k$~functions
$\calM_i \colon \calE(\vec{G}) \rightarrow {\bf C}$, one for each edge
$\vect{a_{2i-1}a_{2i}} \in \calE(\vec{H})$; each $\calM_i$ maps edges
of $\calE(\vec{G})$ to complex roots of unity.  These functions are
defined in such a way that they can ``recognize'' whether a $k$-tuple
of edges $\vec{T} = (\vect{v_1 v_2}, \dots, \vect{v_{2k-1} v_{2k}})$
in $\vec{G}$ forms a copy of $\vec{H}$ with each
$\vect{a_{2i-1}a_{2i}}$ mapping to $\vect{v_{2i-1}v_{2i}}$.  In
particular, if $\vec{T}$ does form such a copy, then the expected
value (over all permissible choices of the maps $\calM_i$) of
$\prod_{i=1}^k \calM_i(\vect{v_{2i-1}v_{2i}})$ is a non-zero constant;
otherwise, the expected value is zero.  Then, as the edges stream by,
the $k$ values $\calZ_i = \sum_{\vect{vw} \in \calE(\vec{G})}
\calM_i(\vect{vw})$ are computed.  Finally, when the stream ends, the
estimate of $\#H$ is given by $\prod_{i=1}^k \calZ_i$ multiplied by an
appropriate constant.

The key to the algorithm is how to define the functions $\calM_i$
so that they can recognize when $\vec{T}$ forms a copy of $\vec{H}$.
Each of these functions $\calM_i$ has two parts: one part is meant
to recognize when $\vec{T}$ forms a homomorphic image of $\vec{H}$,
and the other part is meant to recognize when the $t$ vertices of
this homomorphic image are distinct. In this paper, we do not use
the second part; we use a different method to ensure that the $t$
vertices are distinct.  We therefore omit the second part from
our description, keeping in mind that this description differs
somewhat from the one in~\cite{KMSS}.  For each vertex $b \in H$,
we define a hash function $\calX_b \colon \calV(G) \rightarrow {\bf C}$,
which maps vertices of $G$ to complex $\deg(b)^{\rm th}$ roots of unity,
where $\deg(b)$ is the degree of $b$ in $H$.  Then $\calM_i(\vect{vw})$
is defined to be $\calX_{a_{2i-1}}(v) \calX_{a_{2i}}(w)$.  It is
not difficult to see that $\prod_{i=1}^k \calM_i(\vect{v_{2i-1}v_{2i}})$
has expected value 1 if $\vec{T}$ forms a homomorphic image of $\vec{H}$
with each $\vect{a_{2i-1}a_{2i}}$ mapping to $\vect{v_{2i-1}v_{2i}}$;
otherwise, it has expected value 0.

We can now describe our contributions to this problem.
We present three modifications to the [KMSS]-algorithm, which
can be used separately or together to reduce the storage and
update time per edge.  First, we introduce a different method
for ensuring that we count only those homomorphic images of
$\vec{H}$ that have $t$ distinct vertices.  We do this by
assigning colors to the vertices of $G$.  Assuming there are
$C$ colors, we subdivide each sum $\calZ_i$ into $C^2$
different sums, one for each pair of colors.  For instance,
there might be a red-blue sum
$$\calZ_i^{\rm red, blue} = \sum_{\begin{subarray}{c}
                {v \ {\rm  red}}\\
                {w \ {\rm blue}}\end{subarray}} \calM_i(\vect{vw}) \, .$$
There might also be analogous blue-green sums and green-red sums,
and if we were counting triangles, then
$$\calZ_1^{\rm red, blue}\calZ_2^{\rm blue, green}
\calZ_3^{\rm green, red}$$ would give an estimate for the number of
triangles whose three vertices were respectively red, blue, and green.
This allows us to count only homomorphic images whose vertices all
have different colors, which in turn ensures that the vertices are
all distinct.  However, making sure the vertices are distinct is
not the primary reason we use colors.  The primary reason is that
it dramatically reduces the variance.

For our second modification, rather than defining one hash
function $\calX$ for each vertex of $H$, we define one for each
half-edge of $H$, with the condition that for any vertex $v$ of $G$
and $b$ of $H$, the product $\prod_h \calX_h(v) = 1$, where the
product is taken over all half-edges $h$ in $H$ that are incident to
$b$.  This too reduces the variance of each estimate.

For the third modification, rather than using hash functions
$\calX$ that map vertices to roots of unity, we use hash functions that
map vertices to diagonal $d$-by-$d$ matrices.  Each position along the
diagonal of the matrix more-or-less gives a separate estimate of $\#H$,
so in some sense, this is almost equivalent to making $d$ independent
estimates.  The difference is that, when an edge streams by, instead
of updating each $\calZ_i$ for $d$ different estimates, we only have
to update each $\calZ_i$ for one matrix of estimates.  This lets us
reduce the update time per edge approximately by a factor of $d$.

This paper is organized as follows.  In Section~\ref{sec:alg},
we describe our modified version of the [KMSS]-algorithm and prove
that it gives an unbiased estimate of $\#H$.  In Section~\ref{sec:var},
we bound the variance of our estimate.  In Section~\ref{sec:discussion},
we compare the storage and update time of our algorithm to that of
the original algorithm.

The authors would like to thank Kyle Hofmann, Anthony Gamst,
and Eric Price for many helpful conversations.
\section{Description of Algorithm}
\label{sec:alg}

In this section, we describe our algorithm and show that it gives
an unbiased estimate of $\#H$.  We only explain how to use the
algorithm to produce a single estimate of $\#H$, but in order
to get a more accurate estimate of $\#H$, we would compute
many such estimates and take their average.

Fix some small graph $H$.  We assume throughout the paper that $H$
is connected and has no leaves.  Let $t$ and $k$ respectively denote
the number of vertices and edges in $H$.  Arbitrarily assign directions
to the edges of $H$, and call the resulting directed graph $\vec{H}$.
We assume that the $t$ vertices of $H$ are labeled $1,\dots,t$, and
the $k$ edges are $\vect{a_1 a_2}, \dots, \vect{a_{2k-1} a_{2k}}$,
where each $a_i \in \{1,\dots, t\}$.  $\vec{H}$ has $2k$ half-edges,
which we call $h_1,\dots,h_{2k}$, where $h_{2i-1}$ and $h_{2i}$ are
respectively the two halves of $\vect{a_{2i-1} a_{2i}}$.  In particular,
each $h_j$ is incident to $a_j$. For $b \in \calV(\vec{H})$ define
$\Gamma(b) = \{i : a_i = b\}$.  In other words, $\Gamma(b)$ tells
which half-edges are incident to $b$.  Figure~\ref{fig:sampleH}
illustrates an example where $t=4$ and $k=5$.

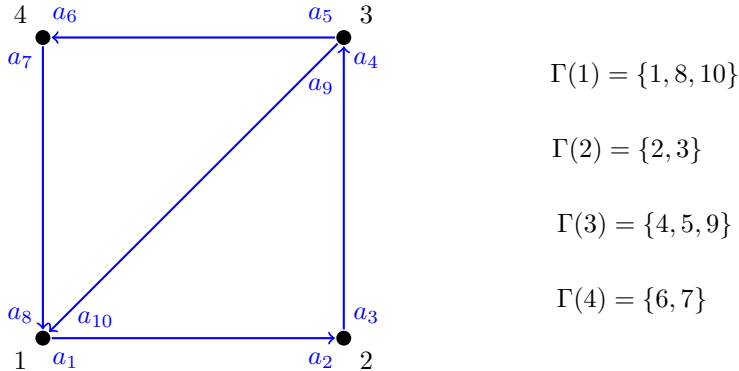
\begin{figure}
\begin{center}
\begin{tikzpicture}

\draw( 1.0, 1.0) node[circle, inner sep=2pt,fill=black,line width=1pt] (1) {};
\draw( 1.0, 5.0) node[circle, inner sep=2pt,fill=black,line width=1pt] (2) {};
\draw( 5.0, 5.0) node[circle, inner sep=2pt,fill=black,line width=1pt] (3) {};
\draw( 5.0, 1.0) node[circle, inner sep=2pt,fill=black,line width=1pt] (4) {};
\draw[thick,draw=blue,<-] (1) -- (2);
\draw[thick,draw=blue,<-] (2) -- (3);
\draw[thick,draw=blue,<-] (3) -- (4);
\draw[thick,draw=blue,<-] (4) -- (1);
\draw[thick,draw=blue,<-] (1) -- (3);
\draw( .7, .7) node{$1$};
\draw(5.3, .7) node{$2$};
\draw(5.3,5.3) node{$3$};
\draw( .7,5.3) node{$4$};
\draw(1.3, .7) node {\color{blue} $a_1$};
\draw(4.7, .7) node {\color{blue} $a_2$};
\draw(5.3,1.3) node {\color{blue} $a_3$};
\draw(5.3,4.7) node {\color{blue} $a_4$};
\draw(4.7,5.3) node {\color{blue} $a_5$};
\draw(1.3,5.3) node {\color{blue} $a_6$};
\draw( .7,4.7) node {\color{blue} $a_7$};
\draw( .7,1.3) node {\color{blue} $a_8$};
\draw(4.7,4.35) node {\color{blue} $a_9$};
\draw(1.7,1.25) node {\color{blue} $a_{10}$};

\draw(9.0, 4.5) node {$\Gamma(1) = \{1,8,10\}$};
\draw(8.77,3.5) node {$\Gamma(2) = \{2,3\}$};
\draw(9.0, 2.5) node {$\Gamma(3) = \{4,5,9\}$};
\draw(8.83,1.5) node {$\Gamma(4) = \{6,7\}$};

\end{tikzpicture}
\end{center}
\caption{Example of $\vec{H}$ and $\Gamma(1),\dots, \Gamma(4)$.}
\label{fig:sampleH}
\end{figure}

For each vertex $b \in \calV(\vec{H})$, select an arbitrary element
$i \in \Gamma(b)$, and call $h_i$ the {\em distinguished} half-edge
at $b$.  Observe that there are $2k$ half-edges in $\vec{H}$, of
which $t$ are distinguished and $2k-t$ are not.
\subsection{The Functions $\calX_i$}
\label{sec:calX}

The [KMSS]-algorithm uses hash functions $\calX$ that map vertices
of $G$ to complex roots of unity.  Here we define similar functions,
but there are two differences. First, instead of having one function
$\calX$ for each vertex of $H$, we have one for each half-edge of $H$.
Second, we allow the more general setting where the co-domain of each
$\calX$ is a group of diagonal matrices.

Let $\calG$ be any finite group of diagonal matrices with
the property that the average of the elements of $\calG$
(i.e., $\sum_{g \in \calG} g/|\calG|$) is the zero matrix.
Note that since $\calG$ consists of diagonal matrices, it
is abelian.  We use $d$ to denote the dimension of the matrices
in $\calG$.  We are primarily interested in two types of groups
$\calG$.  In the first type, $d = 1$, and the elements of $\calG$
are the complex $r^{\rm th}$ roots of unity, for some $r \geq 2$.
In that case, the matrices can be viewed as complex numbers and
are therefore equivalent to what's used in the [KMSS]-algorithm.
For the second type of $\calG$, $d \geq 2$.  Let $\omega = e^{2\pi i/d}$,
and let $M$ be the square diagonal matrix that has $1,\omega,\omega^2,
\dots,\omega^{d-1}$ along the diagonal.  Then $\calG$ is the group
generated by $M$ and $-I$ (where $I$ is the $d$-dimensional identity
matrix); thus $\calG$ has $2d$ elements: $\pm I, \pm M, \pm M^2, \dots,
\pm M^{d-1}$.  In this paper, we focus on those two types of $\calG$,
but we remark that there are other $\calG$ that satisfy the given
conditions; e.g., diagonal matrices whose diagonal entries
are all $\pm 1$.  The entire discussion in this section
applies to any such $\calG$; in particular, our algorithm
gives an unbiased estimate of $\#H$ for any such $\calG$.
However, the discussion of the variance in the next section
applies only to these two specific choices of $\calG$.

Fix any such group $\calG$, and for each $1 \leq i \leq 2k$,
define a hash function $\calX_i \colon \calV(G) \rightarrow \calG$.
If $h_i$ is a non-distinguished half-edge of $\vec{H}$, then
for each $v \in \calV(G)$, the value $\calX_i(v)$ is a random
element of $\calG$, and the functions $\calX_i$ for non-distinguished
$h_i$ are chosen independently and uniformly from a family of
$4k$-wise independent hash functions.  If $h_i$ is the distinguished
half-edge at $b$, then $\calX_i(v)$ is defined by
$$\calX_i(v) = \prod_{j \in \Gamma(b), j \neq i} \calX_j(v)^{-1} \, .$$
If $i$ is the only element of $\Gamma(b)$, then $\calX_i(v)=I$.
Observe that this definition of $\calX_i$ ensures that for any
vertex $b$ of $H$ and any $v \in \calV(G)$,
$$\prod_{j \in \Gamma(b)} \calX_j(v) = I\, .$$

\begin{lemma}
\label{lem:one_vertex}
Let $b \in \{1,\dots,t\}$ be any vertex of $\vec{H}$, and suppose
its degree is~$\delta$. Suppose $\Gamma(b) =  \{i_1,\dots,i_{\delta}\}$;
i.e., $h_{i_1},\dots,h_{i_{\delta}}$ are the  half-edges of $\vec{H}$
incident to $b$. Let $v_1,\dots,v_{\delta}$ be any $\delta$
not-necessarily-distinct vertices of $G$.  Then
$\calX_{i_1}(v_1) \cdots \calX_{i_{\delta}}(v_{\delta})$ is
equal to $I$ if $v_1 = \dots = v_{\delta}$, and otherwise it
is a uniformly random element of $\calG$.
\end{lemma}

\begin{proof}
If $\delta=1$, then the result is clearly true, so assume $\delta>1$.
Assume without loss of generality that the distinguished half-edge at
$b$ is $h_{i_{\delta}}$.  Then by definition,
$$\calX_{i_{\delta}}(v_{\delta}) = \calX_{i_1}(v_{\delta})^{-1} \cdots
                  \calX_{i_{{\delta}-1}}(v_{\delta})^{-1} \, ,$$
so
\begin{equation}
\label{eq:one_vertex}
\prod_{j=1}^{\delta} \calX_{i_j}(v_j) = \prod_{j=1}^{{\delta}-1}
                 \calX_{i_j}(v_j) \calX_{i_j}(v_{\delta})^{-1} \, . 
\end{equation}
If $v_1=\dots=v_{\delta}$, then~\eqref{eq:one_vertex} is equal to $I$.
Now assume that some $v_j \neq v_{\delta}$.  Then for that~$j$,
$\calX_{i_j}(v_j) \calX_{i_j}(v_{\delta})^{-1}$
is the quotient of two independent uniformly random elements of $\calG$,
and is thus a uniformly random element of $\calG$.  Also, none of $h_{i_1},
\dots, h_{i_{{\delta}-1}}$ are distinguished, so
$$
\left(\calX_{i_1}(v_1) \calX_{i_1}(v_{\delta})^{-1}\right), \dots,
\left(\calX_{i_{\delta-1}}(v_{\delta-1})
      \calX_{i_{\delta-1}}(v_{\delta})^{-1}\right)
$$
are independent for all $v_j \ne v_\delta$, and the rest are $I$.
Since at least one is uniformly random, their product is as well.
\end{proof}
\subsection{The Functions $\calM_i$}

Let $\vec{G}$ be the directed graph obtained by replacing each edge $vw$
of $G$ by two directed edges, $\vect{vw}$ and $\vect{wv}$.  Each time an
edge of $G$ streams by, treat it as two directed edges of $\vec{G}$.
From now on, we use $m$ to refer to the number of edges in~$\vec{G}$.
Arguably, we should use $2m$; however, $m$ will be more convenient,
and the factor of 2 will be irrelevant to all of our conclusions,
which use $O()$ notation.

For each edge $\vect{a_{2i-1}a_{2i}}$ of $\vec{H}$, define a function
$\calM_i \colon \calE(\vec{G}) \rightarrow \calG$ by
$$\calM_i(\vect{vw}) = \calX_{2i-1}(v) \calX_{2i}(w) \, .$$
For any $k$-tuple
$\vec{T} = (\vect{v_1 v_2}, \dots, \vect{v_{2k-1} v_{2k}})$
of (not necessarily distinct) edges in $\calE(\vec{G})$, define
$$\calQ(\vec{T}) = \prod_{i=1}^k \calM_i(\vect{v_{2i-1} v_{2i}}) 
                 = \prod_{j=1}^{2k} \calX_j(v_j) \, ,$$
and for each vertex $b \in \calV(\vec{H})$, define
$$\calP_b(\vec{T}) = \prod_{j \in \Gamma(b)} \calX_j(v_j) \, .$$
Since every half-edge of $\vec{H}$ is in exactly one of the sets $\Gamma(b),$
we have
$$\calQ(\vec{T}) = \prod_{b \in \calV(\vec{H})} \calP_b(\vec{T}) \, .$$
The function $\calQ$ will in a sense ``test'' whether $\vec{T}$ forms
a copy of $\vec{H}$.

\begin{lemma}
\label{lem:k_edges}
Let $\vec{T} = (\vect{v_1 v_2}, \dots, \vect{v_{2k-1} v_{2k}})$ be any
$k$-tuple of edges of $\vec{G}$.  Suppose $f\colon\calE(\vec{H}) \rightarrow
\calE(\vec{G})$ sends $\vect{a_{2i-1} a_{2i}}$ to $\vect{v_{2i-1} v_{2i}}$
for each $i$.  If $f$ induces a homomorphism from $\vec{H}$ to $\vec{G}$,
then $\calQ(\vec{T}) = I$.  If $f$ does not induce such a homomorphism,
then $\calQ(\vec{T})$ is a uniformly random element of~$\calG$.
\end{lemma}
\begin{proof}
Suppose $f$ induces a homomorphism from $\vec{H}$ to $\vec{G}$.
Let $b \in \{1,\dots,t\}$ be any vertex of $H$, and suppose the
homomorphism sends $b$ to $w$.  Suppose $\Gamma(b) =  \{j_1,\dots,j_d\}$;
i.e., $h_{j_1}, \dots, h_{j_d}$ are the half-edges of $\vec{H}$ that are
incident to $b$.  Then $v_{j_1},\dots, v_{j_d}$ must all be equal to $w$.
By Lemma~\ref{lem:one_vertex},
$\calX_{j_1}(v_{j_1}) \cdots \calX_{j_d}(v_{j_d}) =~I$.
Equivalently, $\calP_b(\vec{T}) = I$.  This is true for every
$b \in \calV(\vec{H})$, so
$$\calQ(\vec{T}) = \prod_{b \in \calV(\vec{H})} \calP_b(\vec{T}) = I.$$

Now suppose $f$ does not induce such a homomorphism.  Then there must
be some vertex $b$ of $H$ such that, if $\Gamma(b) =  \{j_1,\dots,j_d\}$,
then the vertices $v_{j_1},\dots, v_{j_d}$ are not all equal.  Thus by
Lemma~\ref{lem:one_vertex}, $\calX_{j_1}(v_{j_1}) \cdots \calX_{j_d}(v_{j_d})$
is a uniformly random element of~$\calG$, i.e., $\calP_b(\vec{T})$ is a
uniformly random element.  $\calP_b(\vec{T})$ is independent of
$\calP_c(\vec{T})$ for any other $c \in \calV(\vec{H})$, so
$\prod_{c \in \calV(\vec{H})} \calP_c(\vec{T})$ is also a uniformly
random element; i.e., $\calQ(\vec{T})$ is a uniformly random element.
\end{proof}
\subsection{Coloring Vertices}

Fix some number of colors $C \geq t$.  For the purposes of bounding
the variance, we will later assume that the maximum degree of any
vertex of $G$ is $\leq m^{1/2-\alpha}$ and then set $C = \min(m^{1/3},
m^{2\alpha})$; however, here $C$ may take any value $\geq t$.
Define a hash function $\calC \colon \calV(G) \rightarrow \{1,\dots,C\}$
that assigns a color to each vertex of $G$.  For each vertex $v$,
$\calC(v)$ is a uniformly random color, and $\calC$ is chosen
uniformly at random from a family of $4k$-wise independent hash
functions.

Consider functions $f\colon \calE(\vec{H}) \rightarrow \calE(\vec{G})$.
There are $m^k$ such functions, but we want to find only the ones that
map $\vec{H}$ isomorphically onto its image.  Suppose that $f$ maps
the edges $\vect{a_1a_2},\dots,\vect{a_{2k-1}a_{2k}}$ to the edges
$\vect{v_1v_2},\dots, \vect{v_{2k-1}v_{2k}}$ respectively.
Then for any vertex $b \in \calV(\vec{H})$,
all of the vertices $\{a_i : i \in \Gamma(b)\}$ are equal to $b$;
i.e., they're all the same vertex.  Therefore, a necessary condition
for $f$ to induce an isomorphism is that all the vertices
$\{v_i : i \in \Gamma(b)\}$ are the same vertex.  In particular,
a necessary condition is that all the vertices $\{v_i : i \in \Gamma(b)\}$
have the same color.  Thus we say that either the map $f$ or the $k$-tuple
of edges $\vec{T} = (\vect{v_1v_2},\dots, \vect{v_{2k-1}v_{2k}})$ is
{\em color-compatible} if for every $b \in \calV(\vec{H})$, all the vertices
$\{v_i : i \in \Gamma(b)\}$ have the same color.  More specifically,
for any ordered $t$-tuple of colors $(c_1,\dots,c_t)$, we say that
$\vec{T}$ is {\em $(c_1,\dots,c_t)$-compatible} if for every
$b \in \calV(\vec{H})$, all the vertices $\{v_i : i \in \Gamma(b)\}$
have color $c_b$, or equivalently, if $\calC(v_i) = c_{a_i}$ for
every $1 \leq i \leq 2k$.  Thus $\vec{T}$ is color-compatible if
there exists a $t$-tuple $(c_1,\dots,c_t)$ such that $\vec{T}$
is $(c_1,\dots,c_t)$-compatible.  Furthermore, if $\vec{T}$ is
$(c_1,\dots,c_t)$-compatible and the $t$ colors $c_1,\dots,c_t$
are distinct, then we will say that $\vec{T}$ is {\em distinctly
color-compatible}.

As we saw in Lemma~\ref{lem:k_edges}, $\calQ(\vec{T})$ is equal to $I$
if $\vec{T}$ forms a homomorphic image of $\vec{H}$, and otherwise
is a uniformly random element of $\calG$.  The strategy in~\cite{KMSS}
is basically to compute the sum of $\calQ(\vec{T})$ over all $\vec{T}$.
The sum then has $m^k$ terms and therefore tends to have high variance.
Here, rather than summing over all $\vec{T}$, we will only sum over
distinctly color-compatible $\vec{T}$.  The resulting sum will then
have far fewer terms and therefore tend to have far lower variance.

For colors $c_1,c_2 \in \{1,\dots,C\}$ and $1 \leq i \leq k$, define
\begin{equation}
\label{eq:defZ}
\calZ^{c_1,c_2}_i = \sum_{\begin{subarray}{c}\vect{vw} \in \calE(\vec{G})\ : \\
              \calC(v)=c_1,\ \calC(w)=c_2\end{subarray}} \calM_i(\vect{vw}) \, .
\end{equation}
Thus there are $C^2k$ such sums, and $\calZ^{c_1,c_2}_i$ is the
sum of $\calM_i(\vect{vw})$ over all edges $\vect{vw}$ for which
the color of $v$ is $c_1$ and the color of $w$ is $c_2$.  Also,
define
\begin{equation}
\label{eq:def_subS}
\calS_{(c_1,\dots,c_t)} = \prod_{i=1}^k
\calZ_i^{c_{a_{2i-1}},c_{a_{2i}}} \, .
\end{equation}

We use $E(\,)$ to denote expected value (not to be confused with
$\calE(\,)$, which refers to the edge-set).  We use $\tr(\,)$ to denote
the trace of a matrix.

\begin{lemma}
\label{lem:distinct_color}
For $c_1,\dots,c_t$ distinct, $E(\tr(\calS_{(c_1,\dots,c_t)})/d)$
is equal to the number of $(c_1,\dots,c_t)$-compatible maps
$f \colon \calE(\vec{H}) \rightarrow \calE(\vec{G})$ that induce injective
homomorphisms from $\vec{H}$ to $\vec{G}$.
\end{lemma}
\begin{proof}
From the definitions of $\calS_{(c_1,\dots,c_t)}$ and $\calZ^{c_1,c_2}_i$,
we have
\begin{eqnarray}
\calS_{(c_1,\dots,c_t)} & = &
             \prod_{i=1}^k \calZ_i^{c_{a_{2i-1}},c_{a_{2i}}} \nonumber \\
      & = & \prod_{i=1}^k
             \sum_{\begin{subarray}{c}\vect{vw} \in \calE(\vec{G}),\\
                                         \calC(v)=c_{a_{2i-1}},\\
                                         \calC(w)=c_{a_{2i}}\end{subarray}}
                                  \calM_i(\vect{vw}) \nonumber \\
      & = & \sum_{\begin{subarray}{c}
              \vect{v_1v_2},\dots,\vect{v_{2k-1}v_{2k}},\\
              \calC(v_j) = c_{a_j} {\rm \ for \ } 1\leq j \leq 2k
            \end{subarray}}
            \prod_{i=1}^k \calM_i(\vect{v_{2i-1}v_{2i}}) \nonumber \\
      & = & \sum_{ \begin{subarray}{c}
               \vect{v_1v_2},\dots,\vect{v_{2k-1}v_{2k}} \\
               {\rm is \ }(c_1,\dots,c_t)\mbox{-}{\rm compatible}
            \end{subarray}}
            \prod_{i=1}^k \calM_i(\vect{v_{2i-1}v_{2i}}) \nonumber \\
\label{eq:subS}
      & = & \sum_{ \begin{subarray}{c}
          \vec{T} = (\vect{v_1v_2},\dots,\vect{v_{2k-1}v_{2k}}) \\
               {\rm is \ }(c_1,\dots,c_t)\mbox{-}{\rm compatible}
            \end{subarray}} \calQ(\vec{T}) \, .
\end{eqnarray}
In that last sum, there is one term for every $(c_1,\dots,c_t)$-compatible
map $f \colon \calE(\vec{H}) \rightarrow \calE(\vec{G})$.  Consider any one such term.
By Lemma~\ref{lem:k_edges}, if $f$ does not induce a homomorphism from
$\vec{H}$ to $\vec{G}$, then that term is a uniformly random element
of $\calG$, and, by our assumption on $\calG$, its trace therefore
has expected value 0.  Thus those terms do not contribute to
$E(\tr(\calS_{(c_1,\dots,c_t)}))$.  If $f$ does induce such a homomorphism,
then by Lemma~\ref{lem:k_edges}, that term is equal to $I$, so it
contributes $d$ to the trace of $\calS_{(c_1,\dots,c_t)}$.
Thus $E(\tr(\calS_{(c_1,\dots,c_t)})/d)$ is equal to the number
of $(c_1,\dots,c_t)$-compatible maps $f$ that induce homomorphisms
from $\vec{H}$ to $\vec{G}$.  Since the colors $c_1,\dots,c_t$ were assumed
to be distinct, any such homomorphism sends the vertices of $\vec{H}$ to
vertices of $\vec{G}$ with different colors and is therefore injective.
\end{proof}

Define
$$
\calS = \sum_{\begin{subarray}{c}(c_1,\dots,c_t)\\{\rm distinct}\end{subarray}}
                                        \calS_{(c_1,\dots,c_t)} \, .
$$

\begin{theorem}
\label{thm:alg1}
$$E\left( \frac{C^t}{C(C-1)\cdots(C-t+1)} \cdot
          \frac{\tr(\calS)}{d \cdot \auto(H)} \right) = \#H \, ,$$
          where $\auto(H)$ is the number of automorphisms of $H.$
\end{theorem}

\begin{proof}
By Lemma~\ref{lem:distinct_color}, if $c_1,\dots,c_t$ are distinct colors,
then $\tr(\calS_{(c_1,\dots,c_t)})/d$ gives an unbiased estimate of
the number of $(c_1,\dots,c_t)$-compatible maps $f\colon\calE(\vec{H}) \rightarrow
\calE(\vec{G})$ that induce injective homomorphisms from $\vec{H}$ to $\vec{G}$,
i.e., the number of injective homomorphic images of $\vec{H}$ in
$\vec{G}$ whose vertices have colors $c_1,\dots,c_t$ respectively.  Summing
over distinct $c_1,\dots,c_t$, we see that $\tr(\calS)/d$ gives an unbiased
estimate of the number of injective homomorphic images whose vertices have
distinct colors.  The probability that a randomly colored injective homomorphic
image of $\vec{H}$ has distinct colors is $$\frac{C(C-1)\cdots(C-t+1)}{C^t},$$
so we divide by this expression.  Finally, each copy of $H$ gets counted
as $\auto(H)$ different injective homomorphic images, so we divide by
$\auto(H)$.
\end{proof}

Theorem~\ref{thm:alg1} provides the method for counting copies of
$H$.  As the edges stream by, we compute the sums $\calZ^{c_1,c_2}_i$.
In particular, if the edge $\vect{vw}$ streams by, then for each
$1 \leq i \leq k$, we compute $\calM_i(\vect{vw})$ and add it to the
sum $\calZ_i^{\calC(v),\calC(w)}$.  (For an edge-deletion, we subtract
$\calM_i(\vect{vw})$ from $\calZ_i^{\calC(v),\calC(w)}$.)  Once the
data-stream has ended, for every $t$-tuple of distinct colors
$(c_1,\dots,c_t)$, we compute the product $\calS_{(c_1,\dots,c_t)}$
using Equation~\eqref{eq:def_subS}.  Finally, we sum these values
to get $S$, take the trace, and multiply by
$$\frac{C^t}{C(C-1)\cdots(C-t+1) \cdot d \cdot \auto(H)} $$
to get the final estimate.  We refer to this as {\em Algorithm 1}
and summarize the steps in Table~\ref{tab:alg1}.
Observe that after the data-stream ends, we do a potentially large
computation, which could involve computing roughly $C^t$ values
$\calS_{(c_1,\dots,c_t)}$.  There are often, but not always, ways
to do this computation with less than $C^t$ work.  This is discussed
further in Section~\ref{sec:discussion}.

\begin{table}
\begin{center}
\noindent
\fbox{
\parbox[1cm]{12cm}{
\begin{description}
  \item[{\bf Initialize:}] \ \newline
    For $c_1,c_2 \in \{1,\dots,C\}$ and each
    $1 \leq i \leq k$, set $\calZ^{c_1,c_2}_i = 0$.
    \vspace{.2 in}
  \item[{\bf Update:}] \ \newline
    When an edge $\vect{vw}$ streams by, for each $1 \leq i \leq k$, update
    \begin{eqnarray*}
     \calZ_i^{\calC(v),\calC(w)} & \leftarrow &
       \calZ_i^{\calC(v),\calC(w)} + \calM_i(\vect{vw})\,,
       \text{\ \ \ \ for an insertion,} \\
     \calZ_i^{\calC(v),\calC(w)} & \leftarrow &
       \calZ_i^{\calC(v),\calC(w)} - \calM_i(\vect{vw})\,,
       \text{\ \ \ \ for a deletion.}
    \end{eqnarray*}
  \item[{\bf Final Computation:}] \ \newline
    For $(c_1,\dots,c_t)$ distinct, compute
    $$\calS_{(c_1,\dots,c_t)} = \prod_{i=1}^k
                              \calZ_i^{c_{a_{2i-1}},c_{a_{2i}}}\, .$$
    Then compute
    $$ \calS = \sum_{\begin{subarray}{c}(c_1,\dots,c_t)\\
                                      {\rm distinct}\end{subarray}}
                                        \calS_{(c_1,\dots,c_t)} \, .$$
    Output
    $$\left(\frac{C^t}{C(C-1)\cdots(C-t+1)}\right)
      \left(\frac{\tr(\calS)}{d \cdot \auto(H)}\right) \, .$$
\end{description}
}}
\end{center}
\caption{Algorithm 1}
\label{tab:alg1}
\end{table}

In the case where $\calG = \{\pm I, \pm M, \pm M^2, \dots, \pm M^{d-1}\}$
with $d>1$, a very slight modification to Algorithm~1 reduces the
update time per edge by roughly a factor of $d$.  In this modified
algorithm, which we call {\em Algorithm~2}, we do not compute
the sums $\calZ_i^{c_1,c_2}$ until after the data stream has ended.
Instead, we keep counts of how many times each $M^j$ would have contributed
to $\calZ_i^{c_1,c_2}$.  Thus we have a count for each $i,j,c_1,c_2$,
which we call $\Count_{c_1,c_2}(i,j)$.  Suppose that when some edge
$\vect{vw}$ streams by, we compute $\calM_i(\vect{vw})$ and find that
it is equal to $M^j$.  Rather than immediately adding $M^j$ to
$\calZ_i^{\calC(v),\calC(w)}$, we add~1 to $\Count_{\calC(v),\calC(w)}(i,j)$.
(If $\vect{vw}$ is an edge-deletion or if $\calM_i(\vect{vw})$ is
equal to $-M^j$, then we instead subtract~1 from the count.)
Thus, rather than updating $d$ diagonal entries, we update one count,
saving a factor of $d$ in update time.  The storage does not change
much: for each $\calZ_i^{c_1,c_2}$, rather than storing the values
of~$d$ diagonal entries, we store $d$ counts.  After the data
stream ends, we compute each
\begin{equation}
\label{eq:fft}
\calZ_i^{c_1,c_2} = \sum_{j=0}^{d-1} \Count_{c_1,c_2}(i,j) M^j \, .
\end{equation}
Note that Equation~\eqref{eq:fft} can be evaluated using a fast
Fourier transform, though this is unlikely to have much effect
on the overall run time.  The steps of Algorithm 2 are summarized
in Table~\ref{tab:alg2}.

\begin{table}
\begin{center}
\noindent
\fbox{
\parbox[1cm]{12cm}{
\begin{description}
  \item[{\bf Initialize:}] \ \newline
    For all $c_1,c_2 \in \{1,\dots,C\}$, $1 \leq i \leq k$,
    and $1 \leq j \leq d$,
    set $\Count_{c_1,c_2}(i,j) = 0$.
    \vspace{.2 in}
  \item[{\bf Update:}] \ \newline
    When an insertion edge $\vect{vw}$ streams by,
    for each $1 \leq i \leq k$,\newline
    if $\calM_i(\vect{vw}) = M^j$, then increment
    $\Count_{\calC(v),\calC(w)}(i,j)$;\newline
    if $\calM_i(\vect{vw}) = -M^j$, then decrement
    $\Count_{\calC(v),\calC(w)}(i,j)$. \newline
    For a deletion edge, interchange the increment and decrement.
    \vspace{.2 in}
  \item[{\bf Final Computation:}] \ \newline
    For $c_1,c_2 \in \{1,\dots,C\}$ and each
    $1 \leq i \leq k$,
    compute
    $$ \calZ_i^{c_1,c_2} = \sum_{j=0}^{d-1} \Count_{c_1,c_2}(i,j) M^j \, .$$
    For $(c_1,\dots,c_t)$ distinct, compute
    $$\calS_{(c_1,\dots,c_t)} = \prod_{i=1}^k
                              \calZ_i^{c_{a_{2i-1}},c_{a_{2i}}}\, .$$
    Then compute
    $$ \calS = \sum_{\begin{subarray}{c}(c_1,\dots,c_t)\\
                                      {\rm distinct}\end{subarray}}
                                        \calS_{(c_1,\dots,c_t)} \, .$$
    Output
    $$\left(\frac{C^t}{C(C-1)\cdots(C-t+1)}\right)
      \left(\frac{\tr(\calS)}{d \cdot \auto(H)}\right) \, .$$
\end{description}
}}
\end{center}
\caption{Algorithm 2; here, we are using the group
$\calG=\{\pm I, \pm M, \pm M^2, \dots, \pm M^{d-1}\}$.}
\label{tab:alg2}
\end{table}
\section{The Variance}
\label{sec:var}

In this section, we bound the variance of the estimate given
by our algorithm.  Note that the variance is the same whether we
use Algorithm~1 or Algorithm~2, since they produce the same estimate,
so we do not distinguish between the two.  The variance does however
depend on the choice of $\calG$, and our proof only applies when
$\calG$ is either the group of $r^{\rm th}$ roots of unity or the group
$\{\pm I, \pm M, \pm M^2, \dots, \pm M^{d-1}\}$.  In either case,
the variance is a large sum, but most terms in the sum are zero.
In Section~\ref{sec:roots_of_unity}, we give conditions that classify
which terms contribute non-trivially to the sum when $\calG$ is the group
of $r^{\rm th}$ roots of unity.  In Section~\ref{sec:diagonal_matrices},
we do the same when $\calG$ is the group $\{\pm I, \dots, \pm M^{d-1}\}$.
In Section~\ref{sec:bound_var}, we bound the number of terms that
satisfy those conditions, obtaining our bound.

Our estimate of $\#H$ (which is given in Theorem~\ref{thm:alg1}) has variance
\begin{equation}
\label{eq:var}
\left(\frac{C^t}{C(C-1)\cdots(C-t+1) \cdot d \cdot \auto(H)}\right)^2
E\left(\tr(\calS) \tr(\overline{\calS}) \right) - (\#H)^2 \, ,
\end{equation}
where $\overline{\calS}$ denotes the complex conjugate of $\calS$.  We thus
wish to understand the term $E\left(\tr(\calS) \tr(\overline{\calS})\right)$.

From Equation~\eqref{eq:subS},
$$ \calS_{(c_1,\dots,c_t)} = \sum_{ \begin{subarray}{c}
               \vec{T} = (\vect{v_1v_2},\dots,\vect{v_{2k-1}v_{2k}}) \\
               \text{is $(c_1,\dots,c_t)$-compatible}
            \end{subarray}} \calQ(\vec{T}) \, , $$
so
\begin{equation}
\label{eq:S}
  \calS = \sum_{\begin{subarray}{c}
            (c_1,\dots,c_t) \\
            {\rm distinct}\end{subarray}} \ 
          \sum_{ \begin{subarray}{c}
             \vec{T} = (\vect{v_1v_2},\dots,\vect{v_{2k-1}v_{2k}}) \\
             \text{is $(c_1,\dots,c_t)$-compatible}
            \end{subarray}} \calQ(\vec{T}) 
        = \sum_{ \begin{subarray}{c}
             \vec{T} = (\vect{v_1v_2},\dots,\vect{v_{2k-1}v_{2k}}) \ \text{is}\\
             \text{distinctly color-compatible}
           \end{subarray}} \calQ(\vec{T}) 
\, . 
\end{equation}
Thus $\tr(\calS)\tr(\overline{\calS})$ is a sum of terms of the form
\begin{equation}
\label{eq:one_term}
 \tr(\calQ(\vec{T_1})) \tr(\overline{\calQ(\vec{T_2})}) \, .
\end{equation}
In particular, there is one term for every $2k$-tuple of edges
$(\vec{T_1},\vec{T_2})$ for which
$\vec{T_1} = \vect{v_1v_2},\dots,\vect{v_{2k-1}v_{2k}}$ is
distinctly color-compatible and
$\vec{T_2} = \vect{w_1w_2},\dots,\vect{w_{2k-1}w_{2k}}$ is
distinctly color-compatible.  In contrast, for the [KMSS]-algorithm,
the analogous expression for the variance has a term for each
$2k$-tuple of edges regardless of color-compatibility.

For most $2k$-tuples of edges $(\vec{T_1},\vec{T_2})$, the product
\eqref{eq:one_term} has expected value~0 and therefore does not
contribute to the variance.  Here we classify the $2k$-tuples that
do contribute to the variance.  Consider some $2k$-tuple of
edges $(\vec{T_1},\vec{T_2})$, and consider any vertex $b \in \calV(H)$.
We consider three conditions that the $2k$-tuple may or may not
satisfy at $b$:
\begin{description}
  \item[{\bf Condition 1:}] The vertices $\{v_i: i \in \Gamma(b)\}$ are all
    the same, and the vertices $\{w_i: i \in \Gamma(b)\}$ are all the same.
    \vspace{.05 in}
  \item[{\bf Condition 2:}] $v_i = w_i$ for all $i\in \Gamma(b)$.
    \vspace{.05 in}
  \item[{\bf Condition 3:}] There are vertices $x,y \in \calV(\vec{G})$ such
    that for every $i\in \Gamma(b)$, either $v_i=x$ and $w_i=y$, or
    $v_i=y$ and $w_i=x$.
\end{description}
Note that Condition 1 is a special case of Condition 3.
In general, when Condition~1 is satisfied at every vertex of
$\vec{H}$, each of $\vec{T_1}$ and $\vec{T_2}$ forms a homomorphic
image of $\vec{H}$.  In general, when Condition~2 is satisfied at
every vertex of $\vec{H}$, $\vec{T_1}$ is an arbitrary collection
of $k$ edges, and $\vec{T_2} = \vec{T_1}$.

The following lemma turns Conditions~1--3 into conditions on
$\calP_b(\vec{T_1})$ and $\calP_b(\vec{T_2})$.  Those conditions
will later let us characterize which $\tr(\calQ(\vec{T_1}))
\tr(\overline{\calQ(\vec{T_2})})$ contribute to the variance.

\begin{lemma}
\label{lem:conditions}
Suppose $\vec{T} = (\vec{T_1},\vec{T_2})$ is any $2k$-tuple of
edges of $\vec{G}$.
\begin{enumerate}
  \item \label{condA} If $\vec{T}$ satisfies Condition 1 at $b$, then
    $\calP_b(\vec{T_1}) = \calP_b(\vec{T_2}) = I$.
  \item \label{condB} If $\vec{T}$ satisfies Condition 2 at $b$ but not Condition 1,
    then $\calP_b(\vec{T_1}) = \calP_b(\vec{T_2})$, and each is a
    uniformly random element of $\calG$.
  \item \label{condC} If $\vec{T}$ satisfies Condition 3 at $b$ but not Condition 1,
    then $\calP_b(\vec{T_1}) = \calP_b(\vec{T_2})^{-1}$, and each
    is a uniformly random element of $\calG$.
  \item \label{condD} If $\vec{T}$ does not satisfy Condition 1,2, or 3 at $b$,
    then either $\calP_b(\vec{T_1})$ or~$\calP_b(\vec{T_2})$ is a
    uniformly random element of $\calG$ and is independent of the
    other.
\end{enumerate}
\end{lemma}
\begin{proof}
Suppose that $\vec{T_1} = (\vect{v_1v_2},\dots,\vect{v_{2k-1}v_{2k}})$
and $\vec{T_2} = (\vect{w_1w_2},\dots,\vect{w_{2k-1}w_{2k}})$.
If $\vec{T}$ satisfies Condition 1 at $b$, then by Lemma~\ref{lem:one_vertex},
$\calP_b(T_1) = I$ and $\calP_b(T_2) = I$.

Now suppose that Condition 1 is not satisfied at $b$.  Let $h_\delta$ be the
distinguished half-edge at $b$.  Then
$$\calX_{\delta}(v_{\delta}) =
\prod_{i \in \Gamma(b) \setminus \delta} \calX_i(v_{\delta})^{-1} \, ,$$
so
$$\calP_b(\vec{T_1}) = \prod_{i \in \Gamma(b) \setminus \delta}
                    \calX_i(v_i) \calX_i(v_{\delta})^{-1} \, .$$
Similarly,
$$\calP_b(\vec{T_2}) = \prod_{i \in \Gamma(b) \setminus \delta}
                    \calX_i(w_i) \calX_i(w_{\delta})^{-1} \, .$$
Since Condition 1 is not satisfied, either some $v_i \neq v_{\delta}$
or some $w_i \neq w_{\delta}$.  Assume it is the former.  Then
$\calX_i(v_i) \calX_i(v_{\delta})^{-1}$ is a uniformly random
element of $\calG$, and it is independent of
$\calX_j(v_j) \calX_j(v_{\delta})^{-1}$ for all $j \notin \{i,\delta\}$,
since then neither $h_i$ nor $h_j$ is distinguished.  Thus $\calP_b(\vec{T_1})$
is a uniformly random element of $\calG$.  Similarly, if $w_i \neq w_{\delta}$,
then $\calP_b(\vec{T_2})$ is a uniformly random element of $\calG$.

If $\vec{T}$ satisfies Condition 2, then for each $i \in \Gamma(b)$,
$\calX_i(v_i) = \calX_i(w_i)$, so $\calP_b(\vec{T_1}) = \calP_b(\vec{T_2})$.

If $\vec{T}$ satisfies Condition 3, then for each $i \in \Gamma(b)$,
either $v_i = v_{\delta}$ and $w_i = w_{\delta}$, or $v_i = w_{\delta}$
and $w_i = v_{\delta}$.  Either way, $\calX_i(v_i) \calX_i(v_{\delta})^{-1}$
is the inverse of $\calX_i(w_i) \calX_i(w_{\delta})^{-1}$, so
$\calP_b(\vec{T_1}) = \calP_b(\vec{T_2})^{-1}$.

Suppose then that $\vec{T}$ does not satisfy any of the three conditions.
Suppose also that for some $i \in \Gamma(b)$,
one of $v_i$, $w_i$, $v_{\delta}$, and
$w_{\delta}$ differs from the other three.  Suppose the one that differs
is either $v_i$ or $v_{\delta}$. Then $\calX_i(v_i) \calX_i(v_{\delta})^{-1}$
is a uniformly random element of $\calG$, and it is independent of
$\calX_i(w_i) \calX_i(w_{\delta})^{-1}$.  It is also independent of
$\calX_j(v_j) \calX_j(v_{\delta})^{-1}$ and
$\calX_j(w_i) \calX_j(w_{\delta})^{-1}$ for all $j \notin \{i,\delta\}$.
Thus $\calP_b(\vec{T_1})$ is a uniformly random element of $\calG$ and
is independent of $\calP_b(\vec{T_2})$.  Similarly, if $w_i$ or $w_{\delta}$
was the one that differed from the other three, then $\calP_b(\vec{T_2})$
would be uniformly random and independent of $\calP_b(\vec{T_1})$.
Suppose then that for each~$i$, none of $v_i$, $w_i$, $v_{\delta}$,
and $w_{\delta}$ is different from the other three.  If
$v_{\delta} = w_{\delta}$, then Condition 2 must hold; whereas
if $v_{\delta} \neq w_{\delta}$, then Condition 3 must hold.
\end{proof}

\subsection{Variance When $\calG$ Consists of Roots of Unity}
\label{sec:roots_of_unity}

At this point, the discussion splits into two cases depending on whether
$\calG$ is a group of roots of unity or a group of matrices.  Here we
consider the former.  Therefore we fix some integer $r \geq 2$ and let
$\calG$ be the group of 1-by-1 matrices whose entries are $r^{\rm th}$
roots of unity.  Since the matrices are 1-by-1, we treat all matrices
as complex numbers rather than matrices.  Also, since the trace of a
1-by-1 matrix is equal to its entry, we simply remove ``$\tr$'' from any
equations.  Thus the expression~\eqref{eq:var} for variance becomes
\begin{equation}
\label{eq:var1}
\left(\frac{C^t}{C(C-1)\cdots(C-t+1) \cdot \auto(H)}\right)^2
E(\calS \overline{\calS}) - (\#H)^2 \, .
\end{equation}
Since $\calS \overline{\calS}$ is a sum of terms of the form
$\calQ(\vec{T_1}) \overline{\calQ(\vec{T_2})}$, the next theorem
classifies which pairs $(\vec{T_1},\vec{T_2})$ contribute to
$E(\calS \overline{\calS})$.

\begin{theorem}
\label{thm:var}
Let $\vec{T} = (\vec{T_1},\vec{T_2})$ be a $2k$-tuple of edges of $\vec{G}$.
If either of the following hold:
\begin{itemize}
  \item $\vec{T}$ satisfies Condition 1 or 2 for every $b \in \calV(\vec{H})$,
    or
  \item $r=2$, and $\vec{T}$ satisfies Condition 1, 2, or 3 for every
    $b \in \calV(\vec{H})$,
\end{itemize}
then $\calQ(\vec{T_1}) \overline{\calQ(\vec{T_2})} = 1$.  Otherwise,
$$E\left(\calQ(\vec{T_1}) \overline{\calQ(\vec{T_2})}\right) = 0.$$
\end{theorem}
\begin{proof}
We can write $$\calQ(\vec{T_1}) \overline{\calQ(\vec{T_2})}$$ as
$$\prod_{b \in \calV(\vec{H})} \calP_b(\vec{T_1}) \overline{\calP_b(\vec{T_2})}.
$$
If $\vec{T}$ satisfies Condition~1 or~2 at some $b$, then by
Lemma~\ref{lem:conditions}, $\calP_b(\vec{T_1}) = \calP_b(\vec{T_2})$, so
$$\calP_b(\vec{T_1}) \overline{\calP_b(\vec{T_2})} = 
  \calP_b(\vec{T_1}) \overline{\calP_b(\vec{T_1})} =
  \calP_b(\vec{T_1}) / \calP_b(\vec{T_1}) = 1 \, .$$
If $r=2$ and $\vec{T}$ satisfies Condition~3 at some $b$, then
$\calP_b(\vec{T_1}) = \calP_b(\vec{T_2})^{-1}$, so
$$
\calP_b(\vec{T_1}) \overline{\calP_b(\vec{T_2})} = \calP_b(\vec{T_1})^2 = 1 \, .
$$
Thus if either of these two conditions holds at every $b$, then
$$
\prod_{b \in \calV(\vec{H})}\calP_b(\vec{T_1})\overline{\calP_b(\vec{T_2})} = 1 \, .
$$

Suppose now that at some $b$, Conditions~1 and~2 don't hold.  If
Condition~3 holds and $r>2$, then by Lemma~\ref{lem:conditions},
$\calP_b(\vec{T_1}) = \calP_b(\vec{T_2})^{-1}$, and each is a
uniformly random element of $\calG$.  Then
$\calP_b(\vec{T_1}) \overline{\calP_b(\vec{T_2})} = \calP_b(\vec{T_1})^2$,
and since $r>2,$ we have $$E\left(\calP_b(\vec{T_1})^2\right) = 0.$$  Thus
$$E\left(\calP_b(\vec{T_1}) \overline{\calP_b(\vec{T_2})}\right)=0.$$
If instead Condition~3 does {\em not} hold,
then by Lemma~\ref{lem:conditions}, one of $\calP_b(\vec{T_1})$ and 
$\calP_b(\vec{T_2})$ is a uniformly random $r^{\rm th}$ root of unity 
and is independent of the other, so again
$$E\left(\calP_b(\vec{T_1}) \overline{\calP_b(\vec{T_2})}\right)=0.$$ 
In either case, $\calP_b(\vec{T_1}) \overline{\calP_b(\vec{T_2})}$
is independent of $\calP_c(\vec{T_1}) \overline{\calP_c(\vec{T_2})}$
for $c \in \calV(\vec{H}) \setminus b$, so
$$E\left(\calQ(\vec{T_1}) \overline{\calQ(\vec{T_2})}\right) = 0.$$
\end{proof}
\subsection{Variance When $\calG = \{ \pm I, \pm M, \dots, \pm M^{d-1} \}$}
\label{sec:diagonal_matrices}

Now we consider the variance of our estimate in the case where $\calG$
is a group of matrices.  In particular, fix a dimension $d \geq 2$ and
let $\calG$ consist of the matrices $\{ \pm I, \pm M, \dots, \pm M^{d-1} \}$,
where $M$ is the diagonal matrix with entries
$1,\omega, \omega^2, \dots, \omega^{d-1}$, and $\omega = e^{2\pi i/d}$.

The variance of our estimate for $\#H$ is given by Expression~\eqref{eq:var}.
Note, however, that the trace of every element of $\calG$ is real, so we can
dispense with complex conjugation.  Thus the variance becomes
\begin{equation}
\label{eq:var2}
\left(\frac{C^t}{C(C-1)\cdots(C-t+1) \cdot d \cdot \auto(H)}\right)^2
E(\tr(\calS)^2) - (\#H)^2 \, .
\end{equation}
We thus wish to understand the term $E(\tr(\calS)^2)$.
Since $\tr(\calS)^2$ is a sum of terms of the form
$\tr(\calQ(\vec{T_1})) \tr(\calQ(\vec{T_2}))$, the next theorem
classifies how much each pair $(\vec{T_1},\vec{T_2})$ contributes
to $E(\tr(\calS)^2)$.

\begin{theorem}
\label{thm:var2}
Suppose $\vec{T} = (\vec{T_1},\vec{T_2})$ is a $2k$-tuple of edges of $\vec{G}$.
\begin{itemize}
  \item If $\vec{T}$ satisfies Condition~1 for every $b \in \calV(\vec{H})$,
    then $\tr(\calQ(\vec{T_1})) \tr(\calQ(\vec{T_2})) = d^2$.
  \item If $\vec{T}$ satisfies either Condition~1, 2, or~3 at every
    $b \in \calV(\vec{H})$ but not always Condition~1, then
    $$0 < E\left(\tr(\calQ(\vec{T_1})) \tr(\calQ(\vec{T_2}))\right) \leq d.$$
  \item Otherwise, $$E\left(\tr(\calQ(\vec{T_1}))
                           \tr(\calQ(\vec{T_2}))\right) = 0.$$
\end{itemize}
\end{theorem}
\begin{proof}
Suppose that Condition~1, 2, or~3 holds at every $b \in \calV(\vec{H})$.
Recall that
$\calQ(\vec{T_1}) = \prod_{b \in \calV(\vec{H})} \calP_b(\vec{T_1})$,
and $\calQ(\vec{T_2}) = \prod_{b \in \calV(\vec{H})} \calP_b(\vec{T_2})$.
Let $R_1$ denote the product of $\calP_b(\vec{T_1})$ over all
$b \in \calV(\vec{H})$ where Condition~1 holds.  Let $R_2$ denote
the same product at all $b \in \calV(\vec{H})$ where Condition~2 holds,
but not Condition~1.  Let $R_3$ denote the same product over all 
$b \in \calV(\vec{H})$ where Condition~3 holds, but not Condition~1.
(In each case, if the given conditions are not satisfied at any~$b$,
then define $R_i$ to be $I$.)  Thus $\calQ(\vec{T_1}) = R_1R_2R_3$.
By Lemma~\ref{lem:conditions}-\ref{condA}, $R_1 = I$, so $\calQ(\vec{T_1}) = IR_2R_3$.
By Lemmas~\ref{lem:conditions}-\ref{condB}
and~\ref{lem:conditions}-\ref{condC}, $\calQ(\vec{T_2}) = IR_2R_3^{-1}$.
Furthermore, if there is at least one $b$ where Condition~2 (resp. 3)
holds but not Condition~1, then $R_2$ (resp. $R_3$) is uniformly
random.  Finally, since $R_2$ and $R_3$ involve different vertices
of $H$, they are independent.

If $T$ satisfies Condition~1 at every $b \in \calV(\vec{H})$,
then $R_2 = R_3 = I$, so $\calQ(\vec{T_1}) = \calQ(\vec{T_2}) = I$,
and $\tr(\calQ(\vec{T_1})) \tr(\calQ(\vec{T_2})) = d^2$.

If $T$ satisfies Condition~1 or~2 at every $b \in \calV(\vec{H})$
but not always Condition~1, then $R_3=I$, and
$\calQ(\vec{T_1}) = \calQ(\vec{T_2}) = R_2$.
With probability $1/d$, $R_2 = \pm I$, in which case
$\tr(\calQ(\vec{T_1})) \tr(\calQ(\vec{T_2})) = d^2$.
If $R_2$ is not $\pm I$, then $\tr(R_2) = 0$, so
$\tr(\calQ(\vec{T_1})) \tr(\calQ(\vec{T_2})) = 0$.
Thus $$E\left(\tr(\calQ(\vec{T_1})) \tr(\calQ(\vec{T_2}))\right) = d.$$

If $T$ satisfies Condition~1 or~3 at every $b \in \calV(\vec{H})$ 
but not always Condition~1, then $\calQ(\vec{T_1}) = R_3$, and
$\calQ(\vec{T_2}) = R_3^{-1}$.   With probability $1/d$,
$R_3 = \pm I$, in which case
$\tr(\calQ(\vec{T_1})) \tr(\calQ(\vec{T_2})) = d^2$.
If $R_3$ is not $\pm I$, then $\tr(R_3) = 0$, so
$\tr(\calQ(\vec{T_1})) \tr(\calQ(\vec{T_2})) = 0$.
Thus $$E\left(\tr(\calQ(\vec{T_1})) \tr(\calQ(\vec{T_2})\right) = d.$$

Next, suppose $\vec{T}$ satisfies Condition~1, 2, or~3 at every
$b \in \calV(\vec{H})$ but not always Condition~1 or~2, and not
always Condition~1 or~3.  Then $\calQ(\vec{T_1}) = R_2 R_3$
and $\calQ(\vec{T_2}) = R_2 R_3^{-1}$.  If either of $R_2 R_3$
or $R_2 R_3^{-1}$ is not $\pm I$, then it has trace~0, in which
case $\tr(\calQ(\vec{T_1})) \tr(\calQ(\vec{T_2})) = 0$.  Thus
we only need to consider the cases where $R_2 R_3$ and
$R_2 R_3^{-1}$ are both $\pm I$; or equivalently, the
case where $R_2 = \pm R_3$ and $R_2^2 = I$.  This happens
with probability $1/d^2$ if $d$ is odd and $2/d^2$ if $d$
is even.  Thus $$E\left(\tr(\calQ(\vec{T_1})) \tr(\calQ(\vec{T_2})\right)$$
is equal to 1 if $d$ is odd, and 2 if $d$ is even.

Finally, suppose that $\vec{T}$ does not satisfy any of
Conditions~1, 2, or~3 at some vertex $b \in \calV(\vec{H})$.
Then by Lemma~\ref{lem:conditions}-\ref{condD}, one of $\calQ(\vec{T_1})$
and $\calQ(\vec{T_2})$ is a uniformly random element of $\calG$,
and is independent of the other.  Thus
$$E\left(\tr(\calQ(\vec{T_1}) \tr(\calQ(\vec{T_2})\right) = 
 E\left(\tr(\calQ(\vec{T_1})\right) E\left(\tr(\calQ(\vec{T_2}))\right) = 0.$$
\end{proof}
\subsection{Bounding the Variance}
\label{sec:bound_var}

As we saw in Sections~\ref{sec:roots_of_unity} and~\ref{sec:diagonal_matrices},
a $2k$-tuple of edges only contributes to the variance if it is
distinctly color-compatible and satisfies Condition~1, 2, or~3 
at every vertex of $H$.  Now we bound the number of $2k$-tuples
with these properties to get a bound on the variance.

Throughout this section, $\vec{T} = (\vec{T_1},\vec{T_2})$
will denote a $2k$-tuple of edges of~$\vec{G}$, where
$\vec{T_1} = \vect{v_1v_2},\dots,\vect{v_{2k-1}v_{2k}}$ and
$\vec{T_2} = \vect{w_1w_2},\dots,\vect{w_{2k-1}w_{2k}}$.
We continue to refer to the edges of $\vec{H}$ as
$\vect{a_1 a_2}, \dots, \vect{a_{2k-1} a_{2k}}$.
In much of this section, edge-directions will be irrelevant
and will often be ignored.  We refer to the two halves
of the edge $\overline{a_{2i-1},a_{2i}}$ as the ``half-edge
at $a_{2i-1}$'' and the ``half-edge at $a_{2i}$,'' and
similarly for the two halves of $\overline{v_{2i-1},v_{2i}}$
and $\overline{w_{2i-1},w_{2i}}$.  Let $K$ denote the undirected
subgraph of $G$ consisting of the $2k$ edges of $\vec{T}$, ignoring
edge-directions.  If $i \in \Gamma(b)$ (so $a_i=b$), then we say
that $v_i$ and $w_i$ {\em lie over} $b$.  Thus, for instance,
Condition~1 is satisfied at some $b \in H$ if and only if all
$v_i$ that lie over $b$ are equal and all $w_i$ that lie over
$b$ are equal.  If $i \in \Gamma(b)$ and $\vec{T}$ satisfies
Condition~1 (resp. 2 or 3) at $b$, then we'll say also that
$\vec{T}$ satisfies Condition 1 (resp. 2 or 3) at $v_i$ and at $w_i$.

Suppose $b$ and $c$ are distinct vertices of $H$, and suppose
$i \in \Gamma(b)$ and $j \in \Gamma(c)$.  The vertices $v_i$
and $v_j$ need not be distinct; however, if they are not distinct,
then $\vec{T_1}$ cannot be distinctly color-compatible (because
that would require that $v_i$ and $v_j$ get different colors).  And
similarly for $w_i$, $w_j$, and $\vec{T_2}$.  Thus if we are given
$\vec{T}$ but we are not yet given the colors of the vertices, then
we will say that $\vec{T}$ is {\em distinctly colorable} if for all
distinct vertices $b,c \in H$, and for all $i \in \Gamma(b)$ and
$j \in \Gamma(c)$, the vertices $v_i$ and $v_j$ are distinct, as
are the vertices $w_i$ and $w_j$ (though $v_i$ and $w_j$ are not
required to be distinct).  If $\vec{T}$ is not distinctly colorable,
then no matter how colors are assigned to vertices, $\vec{T}$ will
not be distinctly color-compatible and therefore will not contribute
to the variance.

We begin with some lemmas.

\begin{lemma}
\label{lem:distinct_cond23}
Suppose $i \in \Gamma(b)$ and $j \notin \Gamma(b)$, where $b$
is some vertex of $H$.  If $\vec{T}$ is distinctly colorable
and Condition~2 or~3 is satisfied at $b$, but not Condition~1,
then neither $v_i$ nor $w_i$ can be equal to either $v_j$ or $w_j$.
\end{lemma}

\begin{proof}
If Condition~2 holds at $b$, then $v_i = w_i$.  By the definition
of ``distinctly colorable,'' $v_i \neq v_j$ and $w_i \neq w_j$,
so neither $v_j$ nor $w_j$ can equal $v_i = w_i$.

If instead Condition~3 holds at $b$, then there are two vertices
$x$ and $y$ that lie over $b$ in $K$ such that for all $i' \in
\Gamma(b)$, either $v_{i'} = x$ and $w_{i'} = y$ or vice versa.
We may assume that $v_i=x$ and $w_i=y$.  But since Condition~1
does not hold at $b$, there must also be some $i' \in \Gamma(b)$
such that $v_{i'} = y$ and $w_{i'} = x$.  Since $v_i = x = w_{i'}$,
it follows from the definition of ``distinctly colorable'' that
$v_i$ cannot be equal to either $v_j$ or $w_j$, and similarly
for $w_i$.
\end{proof}

Normally, we refer to the edges of $\vec{H}$ as $\vect{a_1 a_2},
\dots, \vect{a_{2k-1} a_{2k}}$; however, in the next lemma, we
will not be concerned with the directions of the edges, so we
will refer to the edges as $\overline{a_{\alpha} a_{\beta}}$,
with the understanding that for some $1 \leq r \leq k$, either
$\alpha=2r-1$ and $\beta=2r$, or vice versa.

\begin{lemma}
\label{lem:lift_walk}
Suppose $W = \overline{a_{\alpha_1} a_{\beta_1}} \dots
\overline{a_{\alpha_s} a_{\beta_s}}$ is a walk in the undirected graph
$H$, and suppose that Condition~1 or Condition~3 holds at every internal
vertex of the walk (i.e., Condition~1 or~3 holds at each vertex
$a_{\beta_i} = a_{\alpha_{i+1}}$ for $1 \leq i < s$).  Then there
is a walk in $K$ from $v_{\alpha_1}$ to either $v_{\beta_s}$ or
$w_{\beta_s}$, and similarly for $w_{\alpha_1}$.
\end{lemma}

\begin{proof}
Let $e_i$ denote the $i^{\rm th}$ edge of $W$; in other words,
$e_i = \overline{a_{\alpha_i} a_{\beta_i}}$.  Then $e_i$ has two
``lifts'' in $K$, namely, $\overline{v_{\alpha_i} v_{\beta_i}}$
and $\overline{w_{\alpha_i} w_{\beta_i}}$.  We will show that each
lift of $e_i$ is adjacent to a lift of $e_{i+1}$, so we will be
able to piece together lifts of the $e_i$'s to get a lift of the
entire walk.

We use induction on the length $s$ of $W$.  The proof is the
same for $v_{\alpha_1}$ and $w_{\alpha_1}$, so we present the proof
just for $v_{\alpha_1}$.  If $s=1$, then $\overline{v_{\alpha_1}
v_{\beta_1}}$ is the required walk.  If $s>1$, then by induction, there
is a walk $U$ in $K$ from $v_{\alpha_1}$ to either $v_{\beta_{s-1}}$ or
$w_{\beta_{s-1}}$.  Since $W$ is a walk, the edges $e_{s-1}$ and $e_s$
are adjacent; in particular, the vertices $a_{\beta_{s-1}}$ and
$a_{\alpha_s}$ are equal. Equivalently, there is some vertex $b$
of $H$ such that $\beta_{s-1},\alpha_s \in \Gamma(b)$.  By assumption,
Condition~1 or~3 holds at $b$, so either $v_{\beta_{s-1}} = v_{\alpha_s}$
and $w_{\beta_{s-1}} = w_{\alpha_s}$, or $v_{\beta_{s-1}} = w_{\alpha_s}$
and $w_{\beta_{s-1}} = v_{\alpha_s}$.  Either way, we can append either
the edge $\overline{v_{\alpha_s} v_{\beta_s}}$ or the edge
$\overline{w_{\alpha_s} w_{\beta_s}}$ to $U$, obtaining a walk from
$v_{\alpha_1}$ to either $v_{\beta_s}$ or $w_{\beta_s}$.
\end{proof}

Although $H$ is connected, $K$ need not be.  For instance, $K$ might
consist of two isomporphic copies of $H$.  In that case, each connected
component of $K$ contains a lift of every edge of $H$.  However, it can
also happen that a connected component of $K$ contains lifts of only some
edges of $H$.  The next lemmas involve the connected components of $K$.
We generally use $J$ to denote a connected component of $K$ and use $J'$
to denote the subgraph of $H$ that lies ``below'' $J$.  Note that $H$
is connected, so $J'$ is not generally a connected component of $H$.

\begin{lemma}
\label{lem:condition2}
Suppose that $\vec{T}$ satisfies Condition~1, 2, or~3 at each vertex
of $H$, and suppose it satisfies Condition 2 at some vertex.  Then
for every $i \in \{1,\dots,2k\}$, the two vertices $v_i$ and $w_i$
are in the same connected component of~$K$.  Furthermore, that
component also contains some vertex at which Condition~2
is satisfied.
\end{lemma}

\begin{proof}
$H$ is connected, so there is a walk that starts with
the half-edge at $a_i$ and ends at a vertex where Condition~2 holds.
We can choose a minimal such walk, in which case it has no internal
vertices where Condition~2 holds.  Suppose the walk ends with the
half-edge at $a_j$.  By Lemma~\ref{lem:lift_walk}, there is a walk
in $K$ from $v_i$ to either $v_j$ or $w_j$, and also a walk from
$w_i$ to either $v_j$ or $w_j$.  But Condition~2 holds at $a_j$,
so $v_j = w_j$.  Thus there is a walk from $v_i$ to $v_j$, and
one from $w_i$ to~$v_j$.  Concatenating them gives a walk from
$v_i$ to $w_i$.  Thus $v_i$ and $w_i$ are in the same component,
and are also in the same component as the vertex $v_j$, at which
Condition~2 holds.
\end{proof}

\begin{lemma}
\label{lem:Kcomp}
Suppose $\vec{T}$ satisfies Condition~1, 2, or~3 at every vertex
of $H$ and satisfies Condition 2 at some vertex of $H$.   Let $J$
be any connected component of $K$, and define $J'$ to be the subgraph
of $H$ consisting of all edges $\overline{a_{2i-1}a_{2i}}$ for which
either $\overline{v_{2i-1}v_{2i}}$ or $\overline{w_{2i-1}w_{2i}}$
is in $J$.  Then $J'$ must contain either
\begin{itemize}
  \item at least two vertices where $\vec{T}$ satisfies Condition~2;
  \item a vertex with degree at least 2 in $J'$ and where $\vec{T}$
    satisfies Condition~2;
  \item a vertex with degree at least 3 in $H$ and where $\vec{T}$ 
    satisfies Condition~1 or~3.
\end{itemize}
\end{lemma}

\begin{proof}
We assumed there is some vertex of $H$ where Condition 2 is satisfied,
so by Lemma~\ref{lem:condition2}, $J$ contains such a vertex, and so
then does $J'$.  If $J'$ contains two such vertices, then we are done,
so assume there is just one.  If that one vertex has degree at least 2
in $J'$, then again we are done, so assume it has degree~1.  By the
Handshaking Lemma, there must be another vertex with odd degree in~$J'$;
and $\vec{T}$ must satisfy either Condition~1 or~3 at that vertex.

If $b$ is any vertex in $J'$, then for some $i \in \Gamma(b)$, the half-edge
at $a_i$ is in~$J'$, so either the half-edge at $v_i$ or the half-edge at
$w_i$ is in~$J$.  If in addition $\vec{T}$ satisfies either Condition~1
or~3 at $b$, then (by the definition of Conditions~1 and~3), for every
$i \in \Gamma(b)$, either the half-edge at $v_i$ or the half-edge at
$w_i$ is in~$J$.  Therefore, for every $i \in \Gamma(b)$, the half-edge
at $a_i$ is in $J'$.  In other words, $b$ has the same degree in $J'$
as in $H$.  We saw in the previous paragraph that some vertex satisfies
either Condition~1 or~3 and has odd degree in $J'$.  It has the same
degree in $H$.  But we assumed that $H$ has no leaves, so it must
have degree at least 3 in $H$.
\end{proof}

\begin{lemma}
\label{lem:degree_cond2}
Suppose $\vec{T}$ satisfies Condition 1, 2, or 3 at every vertex
of $H$ and satisfies Condition 2 at some vertex $b$ of $H$.  Let
$\delta$ denote the degree of $b$ in~$H$.  For each connected
component $J$ of $K$, define $J'$ as in the previous lemma.
Suppose there are $\kappa$ components $J_1,\dots,J_{\kappa}$ of
$K$ that satisfy:
\begin{itemize}
  \item $J_i$ has at most one vertex where Condition~2 holds, and
  \item $b$ has degree at least two in $J_i'$.
\end{itemize}
Then at most $\delta - \kappa$ distinct vertices of $K$ lie over $b$.
\end{lemma}

\begin{proof}
Suppose $\Gamma(b) = \{s_1, \dots, s_{\delta}\}$, so
the vertices $a_{s_1}, \dots, a_{s_{\delta}}$ are all
equal to~$b$.  The vertices that lie above $b$ in $K$ are
$v_{s_1}, \dots, v_{s_{\delta}}$ and
$w_{s_1}, \dots, w_{s_{\delta}}$.  Since Condition~2
holds at $b$, $v_{s_j} = w_{s_j}$ for each $j$, so in
fact the vertices that lie above $b$ in $K$ are just
$v_{s_1}, \dots, v_{s_{\delta}}$.  Consider any
$J_i'$ as defined in the lemma.  Since $b$ has degree
at least two in $J_i'$, at least two of the half-edges
at $a_{s_1}, \dots, a_{s_{\delta}}$ are in $J_i'$.
Assume without loss of generality that the half-edges
at $a_{s_1}$ and $a_{s_2}$ are in $J_i'$.  Then $J_i$
contains the half-edge at either $v_{s_1}$ or $w_{s_1}$
and the half-edge at either $v_{s_2}$ or $w_{s_2}$.
Since $v_{s_1} = w_{s_1}$ and $v_{s_2} = w_{s_2}$,
$J_i$ contains both $v_{s_1}$ and $v_{s_2}$.  But
$J_i$ has at most one vertex where Condition~2 holds,
so $v_{s_1}$ and $v_{s_2}$ must be the same vertex.
Thus each of $J_1,\dots,J_{\kappa}$ contains two of
$v_{s_1}, \dots, v_{s_{\delta}}$ that are equal,
so there can be at most $\delta - \kappa$ that are distinct.
\end{proof}

\begin{lemma}
\label{lem:mCDelta}
Suppose $0 < \Delta \leq m^{1/2 - \alpha}$, where $\alpha > 0$,
and suppose $0 < C \leq \min(m^{2\alpha},m^{1/3})$.  Then
$\Delta^2 \leq m/C$ and $\Delta \leq m/C^2$.
\end{lemma}

\begin{proof}
The first inequality follows from
$\Delta^2 C \leq m^{1 - 2\alpha} m^{2\alpha} \leq m$.
For the second inequality, if $\alpha \geq 1/6$, then
both $C$ and $\Delta$ are at most $m^{1/3}$, so
$\Delta C^2 \leq m$; if instead $\alpha \leq 1/6$, then
$\Delta C^2 \leq  m^{1/2 - \alpha}(m^{2\alpha})^2 = m^{1/2 + 3\alpha} \leq m$.
\end{proof}

\begin{theorem}
\label{thm:count123}
Suppose that the maximum degree $\Delta$ of any vertex in $G$ is
at most $m^{1/2 - \alpha}$, where $\alpha > 0$, and assume
$C \leq \min(m^{2\alpha},m^{1/3})$.  Then the expected number
of distinctly color-compatible $2k$-tuples of edges of $G$
that satisfy either Condition 1, 2, or 3 at every vertex of $H$,
and satisfy Condition 2 at some vertex of $H$ is
$O(m^k/C^{2k-t})$.
\end{theorem}

\begin{proof}
Let $\vec{T}$ and $K$ be as defined above.  There are $O(1)$
possibilities for the isomorphism class of $K$ (i.e., which of
the vertices $v_1,\dots,v_{2k},w_1,\dots,w_{2k}$ are the same),
so it suffices to prove the theorem for an arbitrary isomorphism
class.  Consider then any one such class.  We may assume that it
is distinctly colorable.

The expected number of possibilities for $\vec{T}$ can be computed in
two steps: first count the number of ways to select the vertices of
$\vec{T}$ where colors are ignored, and then find the probability that
when colors are assigned, $\vec{T}$ becomes distinctly color-compatible.
(When we say ``select the vertices of $\vec{T}$,'' we mean, choose a
vertex of $G$ for each $v_i$ and $w_i$ so that the resulting $\vec{T}$
has the assumed isomorphism class.)
To count the number of ways to select the vertices for $\vec{T}$, we
consider one connected component of $K$ at a time.  Let $J$ be some
connected component of $K$.  We can arbitrarily designate any one edge
of $J$ to be the ``first edge.''  Once we designate the first edge,
there are at most $m$ ways to select its two endpoints (since $\vec{G}$
has $m$ edges).  There are then at most $\Delta$ ways to select each
subsequent vertex of $J$, for a total of $m \Delta^{|\calV(J)|-2}$.
Equivalently, we could have arbitrarily
designated any two (not necessarily adjacent) vertices of $J$ to be
the ``first two vertices;'' we could have then pretended that there
were at most $\sqrt{m}$ ways to select each of those two vertices
and at most $\Delta$ ways to select each other vertex of $J$.

For a component $J$, we use the following method to decide which
will be its first two vertices.  Let $J'$ be the subgraph of
$H$ consisting of all edges $\overline{a_{2i-1}a_{2i}}$ for which
either $\overline{v_{2i-1}v_{2i}}$ or $\overline{w_{2i-1}w_{2i}}$
is in $J$ (as in Lemma~\ref{lem:Kcomp}). By Lemma~\ref{lem:condition2},
$J$ has at least one vertex where Condition~2 is satisfied.  We'll
designate that as one of the first two vertices of $J$.  If
there is a second such vertex, then we'll designate it as the other.
If not, if some vertex of $J$ lies above a vertex of $H$ that has
degree at least 3, and where Condition~1 or~3 is satisfied,
then we'll designate that as the other.  Otherwise, we'll
designate any vertex as the other.

We now show that the result is at most $m^k/C^{2k-t}$.  We consider
one vertex~$b$ of $H$ at a time, and compute the factor that the vertices
that lie above $b$ contribute to the result.  Recall that if some vertex
above $b$ was designated as one of the first two vertices in its component,
then it contributes a factor of $\sqrt{m}$, and otherwise contributes
a factor of $\Delta$.  Furthermore, if Condition~2 holds at $b$, and
if there are $d$ vertices that lie above $b$, then they must all receive
the same color, which introduces a factor of $C^{1-d}$.  Note also
that by Lemma~\ref{lem:distinct_cond23}, if $b_i$ and $b_j$ are
two vertices of $H$ where Condition~2 holds, then all the vertices
that lie above $b_i$ are distinct from all the vertices that lie
above $b_j$, and so these factors of $C^{1-d}$ are all independent.
We consider four cases for $b$.  In the first case, Condition~2 holds
at~$b$.  In the other three cases, Condition~1 or~3 holds at~$b$, but
we subdivide these cases based on whether some vertex that lies above~$b$
was designated as a first vertex of its component, and whether~$b$ has
degree $>2$.

First consider the case where Condition~2 holds at $b$.  Let $d$
denote the number of vertices of $K$ that lie above $b$.  There are
at most $\sqrt{m}$ ways to select each of these $d$ vertices, and
they must all receive the same color, so these vertices contribute at most
a factor of $m^{d/2} / C^{d-1}$ to the count.  By Lemma~\ref{lem:degree_cond2},
$d$ is at most $\delta - \kappa$, where $\delta$ is the degree
of $b$ in $H$, and $\kappa$ is the number of connected components
$J$ of $K$ that satisfy: $J$ has at most one vertex where Condition~2
holds, and $b$ has degree at least two in~$J'$.  Thus the contribution
of $b$ to the overall expected value is at most a factor of
\begin{equation}
\label{eq:count_cond2}
\frac{m^{d/2}}{C^{d-1}} \leq
 \frac{m^{(\delta-\kappa)/2}}{C^{\delta-\kappa-1}} =
\left( \frac{m^{\delta/2}}{C^{\delta-1}} \right)
\left( \frac{C}{m^{1/2}}   \right)^{\kappa} \, .
\end{equation}

For the next three cases, suppose that either Condition~1 or~3 holds
at $b$.  Then at most two vertices of $K$ lie above $b$, and by
Lemma~\ref{lem:condition2}, they lie in the same component of $K$.
In the case where neither was designated as one of the first
two vertices of that component, their contribution to the count
is at most a factor of
\begin{equation}
\label{eq:count_cond13_nofirst}
\Delta^2 \leq \frac{m}{C}
         \leq \frac{m}{C} \left(\frac{m^{1/2}}{C} \right)^{\delta-2}
           =  \frac{m^{\delta/2}}{C^{\delta-1}} \, .
\end{equation}
(We used Lemma~\ref{lem:mCDelta} in the first inequality.)

For the last two cases, suppose that one of the vertices that lie above
$b$ {\em was} designated as one of the first two vertices of the component.
First assume $\delta \geq 3$.  Then the contribution of the vertices that
lie above $b$ to the overall expected value is at most a factor of
\begin{equation}
\label{eq:count_cond13_first_3}
m^{1/2} \Delta  \leq \frac{m^{3/2}}{C^2}
                  =  C \left( \frac{m^{1/2}}{C} \right)^3
                \leq C \left(\frac{m^{1/2}}{C} \right)^{\delta}
                  =  \frac{m^{\delta/2}}{C^{\delta-1}} \, .
\end{equation}
(We used Lemma~\ref{lem:mCDelta} in the first inequality.)

Finally, suppose that one of the vertices that lie above $b$ {\em was}
designated as one of the first two vertices of the component, and
$\delta < 3$ (which means that $\delta=2$).  Then the contribution
of the vertices that lie above $b$ to the overall expected value is at
most a factor of
\begin{equation}
\label{eq:count_cond13_first_2}
m^{1/2} \Delta  = \left( \frac{m}{C} \right)
                  \left( \frac{\Delta C}{m^{1/2}} \right)
                = \left( \frac{m^{\delta/2}}{C^{\delta-1}} \right)
                  \left( \frac{\Delta C}{m^{1/2}} \right)
             \leq \left( \frac{m^{\delta/2}}{C^{\delta-1}} \right)
                  \left( \frac{m^{1/2}}{C} \right) \, .
\end{equation}
(We used Lemma~\ref{lem:mCDelta} in the last inequality.)

Observe that in all four cases (Equations~\eqref{eq:count_cond2},
\eqref{eq:count_cond13_nofirst}, \eqref{eq:count_cond13_first_3},
and \eqref{eq:count_cond13_first_2}), the vertex $b$ contributed
a factor of $m^{\delta/2} / C^{\delta-1}$, except that
in \eqref{eq:count_cond2} and \eqref{eq:count_cond13_first_2},
there are additional factors of $\sqrt{m}/C$ or $C/\sqrt{m}$.
We first show that there are at least as many factors of $C/\sqrt{m}$
as $\sqrt{m}/C$.  For the remainder of the proof, we use $\delta(b)$
rather than $\delta$ to denote the degree of $b$, since $b$ will no
longer be clear from context.  There is one factor of $\sqrt{m}/C$
in Equation~\eqref{eq:count_cond13_first_2} for each vertex
$b$ and component $J$ such that:
\begin{itemize}
  \item $\delta(b) < 3$,
  \item $b$ satisfies Condition~1 or~3, and
  \item a vertex $x$ that lies above $b$ was designated as one of the
    first two vertices of $J$.
\end{itemize}
Observe that $J$ can have at most one vertex that satisfies Condition~1
or~3 and was designated as one of the first two vertices, so
$J$ cannot contribute a factor of $\sqrt{m}/C$ for any vertex
besides $b$.  In other words, $J$ contributes at most one factor
of $\sqrt{m}/C$ overall.  Now we'll show that $J$ also
contributes a factor of $C/\sqrt{m}$ to~\eqref{eq:count_cond2}.
Since $x$ was designated as one of the first two vertices
of $J$, we know that $J$ has only one vertex where Condition~2
holds (which, by Lemma~\ref{lem:distinct_cond23}, implies
that $J'$ has only one vertex where Condition~2 holds),
and $J$ cannot have a vertex that lies above a vertex with
degree at least 3 in $H$ and where Condition~1 or~3 holds.  Thus
by Lemma~\ref{lem:Kcomp}, $J'$ must have a vertex with degree at
least 2 (in $J'$) where Condition~2 holds.  Then for that vertex,
$J$ contributes a factor of $C/\sqrt{m}$ to~\eqref{eq:count_cond2}.
Thus there must be at least as many factors of $C/\sqrt{m}$
in~\eqref{eq:count_cond2} as there are factors of $\sqrt{m}/C$
in~\eqref{eq:count_cond13_first_2}.  We can therefore ignore
all such factors; this can only increase the product.  When
we ignore these factors, each vertex $b$ of $H$ contributes
a factor of at most $m^{\delta(b)/2} / C^{\delta(b)-1}$.
Taking the product over $b$ gives
$$ \frac{m^{\sum_b \delta(b)/2}}{C^{\sum_b (\delta(b)-1)}}
                                      = \frac{m^k}{C^{2k-t}} \, .$$
\end{proof}

Next, we consider what happens when no vertex of $\vec{T}$ satisfies
Condition~2.

\begin{lemma}
\label{lem:2comps}
Suppose $\vec{T}$ satisfies Condition~1 or~3 at every vertex
of $H$.  Then $K$ has at most two connected components.
\end{lemma}

\begin{proof}
$H$ is connected, so for any $i$, there is a walk from $a_i$ to $a_1$.
Since Condition~1 or~3 holds at every vertex along the walk, we can apply
Lemma~\ref{lem:lift_walk} to deduce that there is a walk in $K$ from
$v_i$ to either $v_1$ or $w_1$ and also a walk from $w_i$ to either
$v_1$ or $w_1$.  Thus every vertex of $K$ is in the same connected
component as either $v_1$ or $w_1$.
\end{proof}

\begin{lemma}
\label{lem:3cases}
Suppose $\vec{T}$ satisfies Condition~1 or~3 at every vertex of $H$, but
not always Condition~1.  Then at least one of the following must hold:
\begin{itemize}
  \item $H$ has more edges than vertices;
  \item there are at least two vertices of $H$ where
    $\vec{T}$ does not satisfy Condition~1;
  \item $K$ is connected (as an undirected graph).
\end{itemize}
\end{lemma}

\begin{proof}
We assumed that $H$ is connected and has no leaves.  Suppose the first
condition above does not hold (i.e., $H$ has as many vertices as edges).
Then $H$ must be a cycle.  Now suppose that the second condition above
also does not hold, i.e., there is exactly one vertex of $H$ where $\vec{T}$
does not satisfy Condition~1 (and therefore satisfies Condition~3).
We will assume that the edges of $H$ going around the cycle in order
are $\vect{a_1 a_2}, \vect{a_3 a_4}, \dots, \vect{a_{2t-1} a_{2t}}$.
Note that there is no loss of generality in this assumption, because
edge-directions are irrelevant to this lemma.  We can also assume that
the vertex $a_1 = a_{2t}$
is the vertex where Condition~3 is satisfied, but not Condition~1.
Then $v_1 = w_{2t}$, and $w_1 = v_{2t}$.  Since Condition~1 holds
everywhere else, we have $v_{2i} = v_{2i+1}$ and $w_{2i} = w_{2i+1}$
for all $1 \leq i < t$.  Thus $\overline{v_1 v_2}, \overline{v_3 v_4},
\dots, \overline{v_{2t-1} v_{2t}}, \overline{w_1 w_2}, \overline{w_3 w_4},
\dots, \overline{w_{2t-1} w_{2t}}$ is a path that visits every vertex
of $K$, so $K$ is connected.
\end{proof}

\begin{theorem}
\label{thm:count13}
Suppose that the maximum degree $\Delta$ of any vertex in $G$ is
at most $m^{1/2 - \alpha}$, where $\alpha > 0$, and assume
$C \leq \min(m^{2\alpha},m^{1/3})$.  Then the expected number
of distinctly color-compatible $2k$-tuples of edges of $G$
that satisfy either Condition~1 or~3 at every vertex of $H$,
but do not satisfy Condition~1 at every vertex of $H$, is
$O(m^k/C^{2k-t})$.
\end{theorem}

\begin{proof}
There are again $O(1)$ possibilities for the isomorphism class of
$K$ (i.e., which of the vertices $v_1,\dots,v_{2k},w_1,\dots,w_{2k}$
are the same), so it suffices to prove the theorem for an arbitary
isomorphism class.  Assume then that we are given the isomporphism
class of $K$.  We can assume that it is distinctly colorable.

We first count the number of ways to select the vertices of $K$.
By Lemma~\ref{lem:2comps}, $K$ has at most two components.  As in
the proof of Theorem~\ref{thm:count123}, in each component, we can
arbitrarily designate any one edge to be the ``first edge.''  There
are at most $m$ ways to select its two endpoints (since $\vec{G}$ has
$m$ edges), and there are at most $\Delta$ ways to select each
subsequent vertex in the component.  Equivalently, we can arbitrarily
designate any two (not necessarily adjacent) vertices of the component
to be the ``first two vertices;'' we can then pretend that there are
at most $\sqrt{m}$ ways to select each of these two vertices and at
most $\Delta$ ways to select each subsequent vertex.  Since $K$ has
at most $2t$ vertices (because at most two vertices lie above each vertex
of $H$), there is a total of at most $m \Delta^{2t-2}$ possibilities
for the vertices of $K$ if $K$ has one component, and $m^2 \Delta^{2t-4}$
possibilities if $K$ has two.  Note that the number is greater in
the two-component case.

Once the vertices of $K$ are chosen, colors must be assigned in
such a way that the $2k$-tuple of edges is distinctly color-compatible.
Consider any vertex~$b$ of $H$ where Condition~3 (but not Condition~1)
holds.  That means that exactly two vertices $x$ and $y$ lie above
$b$ in $K$, and they must be distinct (or else Condition~1 would
hold).  Color-compatibility requires that $x$ and $y$ be assigned
the same color, which happens with probability $1/C$.  Furthermore,
if there are two vertices $b_i$ and $b_j$ where Condition~3 (but
not Condition~1) holds, and if $x_i$, $y_i$, $x_j$, and $y_j$ are
the corresponding vertices that lie above $b_i$ and $b_j$, then
by Lemma~\ref{lem:distinct_cond23}, $x_i$, $y_i$, $x_j$, and $y_j$
are all distinct.  Thus every vertex where Condition~3 (but not
Condition~1) holds contributes an independent factor of $1/C$
to the count.

Suppose now that $H$ has more edges than vertices (i.e., $k-t \geq 1$).
There are at most $m^2 \Delta^{2t-4}$ ways to select the vertices
of $K$ and a probability of at most $1/C$ that the result is
distinctly color-compatible, so (using Lemma~\ref{lem:mCDelta})
the expected number of $2k$-tuples of edges is at most
$$
\frac{m^2\Delta^{2t-4}}{C}  =   \frac{m^2 \Delta^2 \Delta^{2t-6}}{C} 
             \leq  \frac{m^2}{C} \left(\frac{m}{C^2}\right)^2
                                  \left(\frac{m}{C}\right)^{t-3} 
               =   \frac{m^{t+1}}{C^{t+2}} 
             \leq  \frac{m^{t+(k-t)}}{C^{t+2(k-t)}} = \frac{m^k}{C^{2k-t}}\,.
$$

Next suppose instead that $H$ has at least two vertices where Condition~3
(but not Condition~1) holds.  Then there are at most $m^2 \Delta^{2t-4}$
ways to select the vertices of $K$ and a probability of at most $1/C^2$
that the result is distinctly color-compatible, so the expected number
of $2k$-tuples of edges is at most
$$
\frac{m^2\Delta^{2t-4}}{C^2}  =  \frac{m^2}{C^2} (\Delta^2)^{t-2} 
   \leq  \left(\frac{m^2}{C^2}\right) \left(\frac{m}{C}\right)^{t-2} 
     =   \frac{m^t}{C^t} 
   \leq  \frac{m^{t+(k-t)}}{C^{t+2(k-t)}} = \frac{m^k}{C^{2k-t}}\,.
$$

The only remaining case is where $H$ does not have more edges than vertices
and $H$ has only one vertex where Condition~3 (but not Condition~1) holds.
By Lemma~\ref{lem:3cases}, $K$ is connected, i.e., has only one component.
Then there are at most $m \Delta^{2t-2}$ ways to select the vertices of
$K$ and a probability of at most $1/C$ that the result is distinctly
color-compatible, so the expected number of $2k$-tuples of edges is at most
$$
\frac{m\Delta^{2t-2}}{C}  =  \frac{m}{C} (\Delta^2)^{t-1} 
   \leq  \left(\frac{m}{C}\right) \left(\frac{m}{C}\right)^{t-1} 
     =   \frac{m^t}{C^t}
   \leq  \frac{m^{t+(k-t)}}{C^{t+2(k-t)}} = \frac{m^k}{C^{2k-t}}\,.
$$
\end{proof}

\begin{lemma}
\label{lem:cond1}
If the $2k$-tuple of edges $\vec{T} = (\vec{T_1},\vec{T_2})$ is
distinctly colorable and satisfies Condition~1 at every vertex
of $H$, then $\vec{T_1}$ and $\vec{T_2}$ are each isomorphic to
$\vec{H}$.  
\end{lemma}

\begin{proof}
Suppose $b$ and $c$ are (not necessarily distinct) vertices of $H$,
and suppose $i \in \Gamma(b)$ and $j \in \Gamma(c)$.  Since
Condition~1 holds everywhere, if $b=c$, then $v_i=v_j$.
Since $\vec{T}$ is distinctly colorable, if $b \neq c$,
then $v_i \neq v_j$.  In other words, $v_i=v_j$ if and
only if $b=c$.  Then the edge map that sends each $\vect{v_{2i-1} v_{2i}}$
to $\vect{a_{2i-1} a_{2i}}$ induces an isomorphism between
$\vec{T_1}$ and $\vec{H}$.  The proof for $\vec{T_2}$
is analogous.
\end{proof}

\begin{theorem}
\label{thm:var_main}
Suppose $\calG$ is either the group of $r^{\rm th}$ roots of unity (in which
case $d=1$) or the group $\{\pm I, \pm M, \pm M^2, \dots, \pm M^{d-1}\}$.
Suppose that the maximum degree $\Delta$ of any vertex in $G$ is at most
$m^{1/2 - \alpha}$, where $\alpha > 0$, and assume
$C \leq \min(m^{2\alpha},m^{1/3})$.  Then the estimate for $\#H$
given by Theorem~\ref{thm:alg1} has variance that is
$O((\#H)^2 + m^k/(dC^{2k-t}))$.
\end{theorem}

\begin{proof}
The variance is given by
$$
\left(\frac{C^t}{C(C-1)\cdots(C-t+1) \cdot d \cdot \auto(H)}\right)^2
E\left(\tr(\calS) \tr(\overline{\calS}) \right) - (\#H)^2 \, .
$$
As discussed earlier,
$\tr(\calS)\tr(\overline{\calS})$ is a sum of terms of the form
$\tr(\calQ(\vec{T_1})) \tr(\overline{\calQ(\vec{T_2})})$,
where
$\vec{T_1} = \vect{v_1v_2},\dots,\vect{v_{2k-1}v_{2k}}$ and
$\vec{T_2} = \vect{w_1w_2},\dots,\vect{w_{2k-1}w_{2k}}$ are
each distinctly color-compatible.
By Theorems~\ref{thm:var} and~\ref{thm:var2}, such a term
contributes to $E(tr(\calS)\tr(\overline{\calS}))$ only if
$(\vec{T_1},\vec{T_2})$ satisfies Condition~1, 2,
or~3 at every vertex of $H$.

First consider the $\vec{T} = (\vec{T_1},\vec{T_2})$
that satisfy Condition~1 at every vertex of $H$.  
By Lemma~\ref{lem:cond1}, if $\vec{T}$ satisfies Condition~1 at every vertex
of $H$ and is distinctly colorable, then $\vec{T_1}$ and $\vec{T_2}$
are each isomorphic to $\vec{H}$.  The number of such $\vec{T}$ is then
$O((\#H)^2)$.  By Theorems~\ref{thm:var} and~\ref{thm:var2}, each
such $\vec{T}$ contributes $d^2$ to $E(\tr(\calS) \tr(\overline{\calS}))$
and therefore contributes at most
$$\left(\frac{C^t}{C(C-1)\cdots(C-t+1) \cdot \auto(H)}\right)^2$$
to the variance.  This term is $O(1)$, so the contributions of
these $\vec{T}$ to the variance is $O((\#H)^2)$.

Next consider the $\vec{T}$ that satisfy Condition~1, 2, or~3 at every
vertex of $H$, but not always Condition~1.  By Theorems~\ref{thm:count123}
and~\ref{thm:count13}, the number of such $\vec{T}$ that are distinctly
color-compatible is $O(m^k/C^{2k-t})$.  By Theorems~\ref{thm:var}
and~\ref{thm:var2}, each such $\vec{T}$ contributes at most $d$
to $E(\tr(\calS) \tr(\overline{\calS}))$.
Thus these terms contribute $O(m^k/(dC^{2k-t}))$ to the variance.
\end{proof}
\section{Discussion of Algorithm}
\label{sec:discussion}

In this section, we discuss how our version of the algorithm
compares to the original in terms of storage, update time per
edge, and a one-time calculation.

First consider the case where $\calG$ is the group of $r^{\rm th}$
roots of unity.  As we showed in Theorem~\ref{thm:var_main},
the variance of a single instance of our algorithm is
$O((\#H)^2 + m^k/C^{2k-t})$, so the number of instances
needed to attain a variance of $O((\#H)^2)$ is
\begin{equation}
\label{eq:instances}
O(1 + m^k/((\#H)^2 C^{2k-t})) \, .
\end{equation}
Each instance of our algorithm requires $O(C^2)$ storage, so
the storage needed is $O(C^2 + m^k/((\#H)^2 C^{2k-t-2}))$.
Assuming our goal is to minimize storage, if the first term
in this expression is larger than the second, then we want
to choose a smaller value of $C$ to ensure that
$$C^2 \leq m^k/((\#H)^2 C^{2k-t-2}) \, ,$$
i.e.,
$$C \leq (m^k / (\#H)^2)^{1/(2k-t)} \, .$$
Thus, although we proved Theorem~\ref{thm:var_main} for
any $C \leq \min(m^{2\alpha},m^{1/3})$, the best choice of
$C$ is $\min(m^{2\alpha},m^{1/3},(m^k / (\#H)^2)^{1/(2k-t)})$.
In that case, the number of instances of our algorithm that
we need to perform is $O(m^k/((\#H)^2 C^{2k-t}))$, so the update
time per edge is also $O(m^k/((\#H)^2 C^{2k-t}))$, and the storage
is $O(m^k/((\#H)^2 C^{2k-t-2}))$.  We thus save a factor of
roughly $C^{2k-t-2}$ in storage and $C^{2k-t}$ in update time
over the original algorithm.

Of the two terms in~\eqref{eq:instances}, 1 and
$m^k/((\#H)^2 C^{2k-t})$, if the first is larger, then we
are doing $O(1)$ instances of the algorithm, so the update time
per edge is $O(1)$.  If the second is larger, then we can reduce
the update time per edge by instead letting $\calG$ be the group
$\{\pm I, \pm M, \dots, \pm M^{d-1}\}$ and setting
$d = m^k/((\#H)^2 C^{2k-t})$, but performing $1/d$ times as many
instances of the algorithm.  By Theorem~\ref{thm:var_main},
the variance remains $O((\#H)^2)$.  The storage requirement
also does not change, since we do $1/d$ times as many instances
of the algorithm, but each instance requires $d$ times the storage.
However, now we are performing $O(1)$ instances of the algorithm,
so the update time is $O(1)$.

There is one drawback of our version of the algorithm: when
the stream ends, a potentially large calculation is required.
In particular, we must compute
$$
\sum_{\begin{subarray}{c}(c_1,\dots,c_t)\\{\rm distinct}\end{subarray}}
                                        \calS_{(c_1,\dots,c_t)} \, . $$
This could potentially involve $C^t$ work, although for most $H$, we can
use inclusion-exclusion to perform the calculation more efficiently.
For instance, if $H$ is a 4-cycle with vertices 1,2,3,4 and edges $\vect{12},
\vect{23}, \vect{34}, \vect{41}$,  then we can loop through colors $c_1$
and $c_3$ for vertices~1 and~3.  For each such pair of colors, we can
loop through colors $c_2 \notin \{c_1,c_3\}$ for vertex~2, computing
$$ \sum_{c_2 \notin \{c_1,c_3\}} \calZ_1^{c_1,c_2} \calZ_2^{c_2,c_3} \, . $$
Separately, we can loop through colors $c_4 \notin \{c_1,c_3\}$ for
vertex~4, computing
$$ \sum_{c_4 \notin \{c_1,c_3\}} \calZ_3^{c_3,c_4} \calZ_4^{c_4,c_1} \, .  $$
We can multiply those two sums and then subtract the terms where $c_2 = c_4$:
$$ \sum_{c \notin \{c_1,c_3\}} \calZ_1^{c_1,c} \calZ_2^{c,c_3}
                              \calZ_3^{c_3,c} \calZ_4^{c,c_1} \, .$$
We thus do the computation with $C^3$ work rather than $C^4$ work.
In fact, it is possible to do slightly better: each of the three sums
above can be computed for all $c_1$ and $c_3$ by performing a $C \times C$
matrix multiplication, which can be done using less than $C^3$ work.
It would be unusual for this computation to be a significant issue,
but if it is, then we might want to choose a smaller value of~$C$,
in which case we would not realize the full reduction in storage.
\section{Conclusion}

We have described three modifications to the [KMSS]-algorithm:
we define one hash function $\calX_i$ for each half-edge of $H$
rather than one for each vertex of $H$; we assign colors to the
vertices of $G$ and restrict to distinctly color-compatible~$\vec{T}$;
and we allow matrix-valued hash functions as an alternative to
complex-valued hash functions.  The first two modifications
reduce the variance in each instance of the algorithm, and
therefore reduce the number of instances needed.  This in
turn reduces the required storage and update time per edge.
The third modification reduces only the update time per edge.

Suppose that the maximum degree $\Delta$ of any vertex in $G$
is at most $m^{1/2 - \alpha}$, where $\alpha > 0$, and
suppose $C \leq \min(m^{2\alpha},m^{1/3})$.  For the original
[KMSS]-algorithm, both the storage and update time per edge
are $O(m^k/(\#H)^2)$.  For our algorithm, we have shown that
the update time per edge is $O(1)$, and the storage is
$O(C^2 + m^k/(C^{2k-t-2}(\#H)^2))$, i.e., the storage
has been reduced approximately by a factor of $C^{2k-t-2}$.


\begin{thebibliography}{99}

\bibitem{ADNK}
N. Ahmed, N. Duffield, J. Neville, and R. Kompella.
Graph sample and hold: a framework for big-graph analytics.
{\it KDD 2014}, pp. 1446-1455, 2014.

\bibitem{ADWR}
N. Ahmed, N. Duffield, T. Wilke, R. Rossi, ``On Sampling from Massive
Graph Streams,'' in Proc. VLDB, 1430-1441, 2017.

\bibitem{AR}
N. Ahmed and R. Rossi.
The Network Data Repository with Interactive Graph Analytics and Visualization.
http://networkrepository.com, 2015.

\bibitem{AGM}
K. Ahn, S. Guha, and A. McGregor.
Graph sketches: Sparsification, spanners, and subgraphs.
In {\it Proceedings of the Symposium on Principles of Database Systems (PODS)},
2012, pp. 5-14.

\bibitem{ACIKMS}
U. Alon, D. Chklovskii, S. Itzkovitz, N. Kashtan, R. Milo, S. Shen-Orr.
Network motifs: simple building blocks of complex networks.
{\it Science} 298, no. 5594, pp. 824-827, 2002.

\bibitem{AKK}
S. Assadi, M. Kapralov, S. Khanna.
A simple sublinear-time algorithm for counting arbitrary subgraphs via edge
sampling. In {\it ITCS}, volume 124 of {\it LIPIcs}, pp. 6:1-6:20. Schloss
Dagstuhl - Liebniz-Zentrum fuer Informatik, 2019.

\bibitem{BKS}
Z. Bar-Youssef, R. Kumar, and D. Sivakumar.
Reductions in streaming algorithms with an application to counting triangles
in graphs.
{\it SODA}, pp. 623-632, 2002.

\bibitem{BC}
S. Bera and A. Chakrabarti.
Towards tighter space bounds for counting triangles and other substructures in
graph streams.
In {\it 34th Symposium on Theoretical Aspects of Computer Science (STACS 2017)},
pp. 11:1-11:14, 2017.

\bibitem{BFKP}
L. Bulteau, V. Froese, K. Kutzkov, and R. Pagh.
Triangle counting in dynamic graph streams.
{\it Algorithmica}, vol. 76, no. 1, pp. 259-278, 2016.

\bibitem{BFLMS}
L.S. Buriol, G. Frahling, S. Leonardi, A Marchetti-Spaccamela, and C. Sohler.
Counting triangles in data streams.
{\it PODS}, pp. 253-262, 2006.

\bibitem{CGT}
A. Chakrabarti, P. Ghosh, and J. Thaler.
Streaming verification for graph problems: Optimal tradeoffs and nonlinear
sketches. {\it To appear in RANDOM}, 2020.

\bibitem{CL}
X. Chen and J. Lui.
A unified framework to estimate global and local graphlet counts for streaming
graphs.  In {\it Proceedings of the 2017 IEEE/ACM international Conference on
Advances in Social Networks Analysis and Mining 2017}, pp. 131-138, 2017.

\bibitem{CC}
S. Chu and J. Cheng.
Triangle listing in massive networks and its applications.
In {\it Proceedings of the International Conference on Knowledge Discovery
and Data Mining (SIGKDD)}, 2011, pp. 672-680.

\bibitem{ERSU}
A. Epasto, M. Riondato, L. Stefani, and E. Upfal.
Tri\`{e}st: Counting local and global triangles in fully-dynamic
streams with fixed memory size.
{\it KDD}, pp. 825-834, 2016.

\bibitem{ELS}
D. Eppstein, M. L\"{o}ffler, and D. Strash.
Listing all maximal cliques in large sparse real-world graphs.
{\it ACM J. Exp. Algorithmics}, 18 (2013), pp. 3-1,
https://doi.org/10.1145/2543629.

\bibitem{FHKS}
C. Faloutsos, B. Hooi, J. Kim, K. Shin.
Think before you discard: Accurate triangle counting in graph streams with
deletions. In {\it Joint European Conference on Machine Learning and Knowledge
Discovery in Databases}, pp. 141-157, Springer, Cham, 2018.

\bibitem{HS}
G. Han and H. Sethu.
Edge sample and discard: A new algorithm for counting triangles in large
dynamic graphs.
{\it ASONAM}, pp. 44-49, 2017.

\bibitem{JSP}
M. Jha, C. Seshadri, and A. Pinar.
A space-efficient streaming algorithm for estimating transitivity and triangle
counts using the birthday paradox.
{\it TKDD}, vol. 9, no. 3, pp. 15:1-15:21, 2015.

\bibitem{JG}
H. Jowhari and M. Ghodsi.
New streaming algorithms for counting triangles in graphs.
{\it COCOON}, pp. 710-716, 2005.

\bibitem{KP}
J. Kallaugher and E. Price.
A hybrid sampling scheme for triangle counting.
In {\it Proceedings of the Twenty-Eighth Annual ACM-SIAM Symposium on
Discrete Algorithms}, pp. 1778-1797.
Society for Industrial and Applied Mathematics, 2017.

\bibitem{KKP}
J. Kallaugher, M. Kapralov, and E. Price. 
The sketching complexity of graph and hypergraph counting,
{\it FOCS}, pp. 556-567, 2018.

\bibitem{KMSS}
D. M. Kane, K. Mehlhorn, T. Sauerwald, and H. Sun.
Countng arbitrary subgraphs in data streams.
{\it ICALP}, pp. 598-609, 2012.

\bibitem{KL}
U. Kang and Y. Lim.
MASCOT: memory-efficient and accurate sampling for counting local triangles
in graph streams.
{\it KDD}, pp. 685-694, 2015.

\bibitem{KHP}
Neeraj Kavassery-Parakkat, Kiana Mousavi Hanjani, and A. Pavan.
Improved triangle counting in graph streams: power of multi-sampling.
In {\it 2018 IEEE/ACM International Conference on Advances in Social Networks
Analysis and Mining (ASONAM)}, pp. 33-40. IEEE, 2018.

\bibitem{KMT}
M. Kolountzakis, G. Miller, and C. Tsourakakis.
Triangle Sparsifiers.
{\it J. Graph Algorithims Appl.}
15, no. 6 (2011): 703-726.

\bibitem{KMPT}
M. Kolountzakis, G. Miller, R. Peng, and C. Tsourakakis.
Efficient triangle counting in large graphs via degree-based
vertex partitioning.
{\it Internet Mathematics}, 8(1-2):161-185, 2012.

\bibitem{MMPS}
M. Manjunath, K. Mehlhorn, K. Panagiotou, and H. Sun.
Approximate counting of cycles in streams.
{\it ESA}, pp. 677-688, 2011.

\bibitem{MVV}
A. McGregor, S. Vorotnikova, and H. T. Vu.
Better algorithms for counting triangles in data streams.
In {\it Proc. 35th ACM Symposium on Principles of Database Systems},
pp. 401-411, 2016.

\bibitem{PT}
R. Pagh and C. E. Tsourakakis.
Colorful triangle counting and a MapReduce implementation.
{\it Inf. Process. Lett.}, vol. 112, no. 7, pp. 277-281. 2012.

\bibitem{PTTW}
A. Pavan, K. Tangwongsan, S. Tirthapura, and K. Wu.
Counting and sampling triangles from a graph stream.
{\it PVLDB}, vol. 6, no. 14, pp. 1870-1881, 2013.

\bibitem{SSTZ}
S. Sanei-Mehri, A. Sariy\"{u}ce, S. Tirthapura, and Y. Zhang.
FLEET: butterfly estimation from a bipartite graph stream.
In {\it Proceedings of the 28th ACM International Conference on Information
and Knowledge Management},
pp. 1201-1210, 2019.

\bibitem{T}
C. E. Tsourakakis.
The k-clique densest subgraph problem.
In {\it Proceedings of the International Conference on World Wide Web (WWW)},
2015, pp. 1122-1132,
https://doi.org/10.1145/2736277.2741098.

\end{thebibliography}
\end{document}